\documentclass[11pt]{article}
\usepackage[margin=1in]{geometry}
\usepackage[utf8]{inputenc}

\usepackage{amsmath}
\usepackage{amsthm}
\usepackage{amssymb}
\usepackage{graphicx}
\usepackage{ifthen}
\usepackage{calc}

\usepackage{enumitem}
\usepackage{hyperref}
\usepackage{xspace}

\newtheoremstyle{mythm}%
  {}%
  {}%
  {\itshape}%
  {}%
  {\bfseries}%
  {.}%
  {.5em}%
  {\thmname{#1}~\thmnumber{#2}\ifthenelse{\equal{\thmnote{#3}}{}}{}{~(\thmnote{#3})}}%

\newtheoremstyle{mydefn}%
  {}%
  {}%
  {\upshape}%
  {}%
  {\bfseries}%
  {.}%
  {.5em}%
  {\thmname{#1}~\thmnumber{#2}\ifthenelse{\equal{\thmnote{#3}}{}}{}{~(\thmnote{#3})}}%

\newtheoremstyle{myremark}%
  {}%
  {}%
  {\upshape}%
  {}%
  {\itshape}%
  {.}%
  {.5em}%
  {\thmname{#1}~\thmnumber{#2}\ifthenelse{\equal{\thmnote{#3}}{}}{}{~(\thmnote{#3})}}%

\theoremstyle{mythm}
\newtheorem{theo}{Theorem}[section]
\newtheorem{lem}[theo]{Lemma}

\newtheorem{cor}[theo]{Corollary}
\newtheorem{fact}[theo]{Fact}

\theoremstyle{mydefn}
\newtheorem{defn}[theo]{Definition}

\theoremstyle{myremark}

\theoremstyle{mythm}

\newcounter{claimcounter}

\newlist{caselist}{description}{10}
\setlist[caselist]{font=\itshape\mdseries}

\setenumerate[1]{label=(\arabic*),itemsep=0pt,topsep=1pt}
\newlist{eroman}{enumerate}{2}
\setlist[eroman,1]{label=(\roman*),itemsep=0pt,topsep=1pt}
\setlist[eroman,2]{label=(\alph*), itemsep=0pt}
\newlist{ealph}{enumerate}{1}
\setlist[ealph]{label=(\Alph*),itemsep=0pt,topsep=1pt}

\newcounter{nlistcounter}

\renewcommand{\phi}{\varphi}

\newcommand{\Bigmid}{\;\Big|\;}
\newcommand{\ceil}[1]{\left\lceil#1\right\rceil}
\newcommand{\floor}[1]{\left\lfloor#1\right\rfloor}
\renewcommand{\mathbf}[1]{\textit{\bfseries #1}}

\renewcommand{\hat}[1]{{\widehat{#1}}}
\newcommand{\bigoh}[1]{O(#1)}
\newcommand{\Bigoh}[1]{O\bigl(#1\bigr)}
\newcommand{\setsize}[1]{ |#1| }

\newcommand{\set}[1]{ \{#1\} }

\newcommand{\setdescr}[2]{ \{#1\;\mid\; #2\} }

\newcommand{\defeq}{:=}

\renewcommand{\tilde}[1]{\widetilde{#1}}

\newcommand{\CB}{{\mathcal B}}
\newcommand{\CC}{{\mathcal C}}
\newcommand{\CE}{{\mathcal E}}

\newcommand{\CY}{{\mathcal Y}}
\newcommand{\CX}{{\mathcal X}}

\newcommand{\CM}{{\mathcal M}}

\renewcommand{\vec}[1]{{\boldsymbol #1}}

\newcommand{\bbeta}{{\boldsymbol\beta}}
\newcommand{\bgamma}{{\boldsymbol\gamma}}

\newcommand{\balpha}{{\boldsymbol\alpha}}

\newcommand{\Piso}{{\mathsf P}_{\textup{iso}}}
\newcommand{\LL}{{\mathsf L}}
\newcommand{\Liso}{\LL_{\textup{iso}}}
\newcommand{\Lcsp}{\LL_{\textup{csp}}}
\newcommand{\var}[1]{\left[#1\right]}

\newcommand{\Carb}{{C_{\text{arb}}}}
\newcommand{\Rarb}{{R_{\text{arb}}}}
\newcommand{\CCext}{{{\CC^\ast}}}
\newcommand{\extensionparameter}{e}
\newcommand{\Pcsp}{\mathsf P_{\text{\upshape csp}}}
\newcommand{\Pext}{\Pcsp(\CCext)}
\newcommand{\arityk}{k} %
\newcommand{\domD}{D} %

\newcommand{\relR}{R} 
\newcommand{\degd}{d} %
\newcommand{\refdegd}{d} %

\newcommand{\SP}{{\mathsf P}}
\newcommand{\SQ}{{\mathsf Q}}

\newcommand{\givar}[2]{[#1\mapsto #2]} %
\newcommand{\cspvar}[2]{[#1\mapsto #2]}

\newcommand{\dirH}{H} %

\newcommand{\graphcsp}{G(\CC)}
\newcommand{\graphcsptilde}{\tilde G(\CC)}

\newcommand{\polyf}{f}
\newcommand{\polyg}{g}
\newcommand{\isopi}{\pi} %
\newcommand{\assignI}{I} %

\newcommand{\ZZ}{\mathbb Z}
\newcommand{\QQ}{\mathbb Q}
\newcommand{\dom}{\operatorname{dom}}
\newcommand{\Dom}{\operatorname{Dom}}

\newcommand{\Var}{\operatorname{Var}}

\newcommand{\psolution}{$p$-solution\xspace} %
\newcommand{\psolutions}{$p$-solutions\xspace} %

\newcommand{\vE}{{\overrightarrow E}}
\newcommand{\cl}{\operatorname{cl}}
\newcommand{\rk}{\operatorname{rk}}

\newcommand{\Bd}{\partial}

\usepackage{alphalph}
\numberwithin{equation}{section}

 \newcommand{\mappsi}{\psi}
 \newcommand{\indm}{m}
 \newcommand{\aritykalt}{{\arityk'}}
 \newcommand{\varnormal}[1]{[#1]}
 \newcommand{\Phiiso}{\Phi_{\text{iso}}}
 \newcommand{\Phicsp}{\Phi_{\text{csp}}}
 
\newcommand{\matM}{M}
\newcommand{\vecx}{\vec x}
\newcommand{\vecy}{\vec y}
\newcommand{\vecb}{\vec b}
\newcommand{\dimN}{n}
\newcommand{\dimM}{m}

  \newcommand{\LisoLcsp}{$\Liso^\ell(G,H)$ / $\Lcsp^\ell(\CC)$\xspace}
  \newcommand{\PisoPcsp}{$\Piso^\ell(G,H)$ / $\Pcsp^\ell(\CC)$\xspace}

\allowdisplaybreaks

\begin{document}
\title{Linear Diophantine Equations, Group CSPs, and Graph Isomorphism}
\author{%
\begin{tabular}{c}
Christoph Berkholz\\
\normalsize HU Berlin\\
\normalsize berkholz@informatik.hu-berlin.de
\end{tabular}
\qquad
\begin{tabular}{c}
          \large Martin Grohe\\
          \normalsize RWTH Aachen University\\
          \normalsize grohe@informatik.rwth-aachen.de
\end{tabular}
}
\date{}
\maketitle

\begin{abstract}
In recent years, we have seen several approaches to the graph
isomorphism problem based on ``generic'' mathematical programming or
algebraic (Gr\"obner basis) techniques. For most of these, lower
bounds have been established. In fact, it has been shown that the
pairs of non-isomorphic CFI-graphs (introduced by Cai, F\"urer, and
Immerman in 1992 as hard examples for the combinatorial
Weisfeiler-Leman algorithm) cannot be distinguished by these
mathematical algorithms. A notable exception were the algebraic
algorithms over the field $\mathbb F_2$, for which no lower bound was
known. Another, in some way even stronger, approach to graph
isomorphism testing is based on solving systems of linear Diophantine
equations (that is, linear equations over the integers), which is
known to be possible in polynomial time. So far, no lower bounds for
this approach were known.

Lower bounds for the algebraic algorithms can best be proved in the
framework of proof complexity, where they can be phrased as lower
bounds for algebraic proof systems such as Nullstellensatz or the
(more powerful) polynomial calculus. We give new hard examples for
these systems: families of pairs of non-isomorphic graphs that are
hard to distinguish by polynomial calculus proofs simultaneously over
all prime fields, including $\mathbb F_2$, as well as examples that are hard
to distinguish by the systems-of-linear-Diophantine-equations
approach.

In a previous paper, we observed that the CFI-graphs are closely
related to what we call ``group CSPs'': constraint satisfaction
problems where the constraints are membership tests in some coset of a
subgroup of a cartesian power of a base group ($\mathbb Z_2$ in the case of the
classical CFI-graphs). Our new examples are also based on group CSPs
(for Abelian groups), but here we extend the CSPs by a few non-group
constraints to obtain even harder instances for graph isomorphism.


\end{abstract}

\section{Introduction}

The graph isomorphism problem is famous for its unsolved complexity
status, and despite exciting recent developments in graph isomorphism
testing \cite{Babai16GraphIsomorphism}, a polynomial time algorithm is
not in sight.  Recently, generic mathematical programming and
algebraic techniques applied to graph isomorphism have received
considerable
attention~\cite{atsman13,BerGro15,GroheOttoJournal.2016,mal14,codschsno14,DWWZ.2013}. The
basic idea is to
encode the isomorphism problem for two given graphs $G$ and $H$ into a
system of equalities and inequalities in variables $\givar{v}{w}$ for
vertices $v\in V(G)$ and $w\in V(H)$. The intended meaning of the
variable $\givar{v}{w}$ is to indicate whether $v$ is mapped to $w$
(value $1$) or not (value $0$).  The coding details depend on the
exact algorithmic framework: sometimes we use linear equalities and
inequalities, sometimes we use linear and quadratic equalities, and
sometimes we use additional variables such as
$\var{v_1\mapsto w_1,\ldots,v_\ell\mapsto w_\ell}$ indicating that
$v_i$ is mapped to $w_i$ for $i=1,\ldots,\ell$. Furthermore, we
interpret the equations over different fields and rings. Then we solve
or try to solve the system (using different methods like linear or
semi-definite programming or Gr\"obner bases), which should tell us
whether the given graphs are isomorphic, but not always does. All the
polynomial time algorithms based on this paradigm either correctly
detect that the graphs are non-isomorphic or give no definite
answer. In the former case, we say that the algorithm
\emph{distinguishes} the graphs. Hence to prove that the algorithm is
a \emph{complete} isomorphism test we have to show that it
distinguishes all pairs of non-isomorphic graphs. Not surprisingly,
most of these algorithms have been proved to be incomplete. Somewhat
surprisingly, despite the considerable variation of systems that have
been studied, it has turned out that all these algorithms are very
similar in their distinguishing power. In particular, they all fail to
distinguish the non-isomorphic pairs of CFI-graphs, introduced by Cai,
Fürer, and Immerman~\cite{caifurimm92} to prove that the
Weisfeiler-Leman (WL) algorithm, a combinatorial graph isomorphism
test, is incomplete. The distinguishing power as well as the running
time of all these algorithm is governed by a parameter $\ell$, which
is the degree of the polynomials considered by a Gr\"obner basis
algorithm, the
``level'' in a hierarchy of linear and semidefinite programming
relaxations, or the ``dimension'' of the WL algorithm.  Proving lower
bounds for any of the algorithms means proving lower bounds on the
parameter $\ell$ necessary to distinguish the input graphs. For almost
all of the algorithms, the CFI-graphs yield a lower bound on $\ell$
that is linear in the size $n$ of the input graphs. As the running
time of the algorithms is $n^{\Theta(\ell)}$, these lower bounds not
only show that the polynomial time restrictions of the algorithms are
incomplete isomorphism tests, but are in fact much stronger. There are
two algorithms among those considered in this context whose
incompleteness had not been established. The first is based on solving
systems of linear Diophantine equations (that is, linear equations
over the integers), which is possible in polynomial time (see, for
example, \cite{schri86}). The second is based on the Gr\"obner basis
algorithm over fields of characteristic $2$. We prove lower bounds for
both of these algorithms. Before we explain how these lower bounds are
obtained, let us discuss both algorithms in more detail.

The systems of linear Diophantine equations for two graphs $G,H$ are
obtained as follows. We start from a standard integer linear program
in the variables $\var{v\mapsto w}$ whose nonnegative integral
solutions are the isomorphisms between $G$ and $H$. As the standard
LP-relaxation is fairly weak, we strengthen the system using
lift-and-project methods, specifically the Sherali-Adams hierarchy
\cite{Sherali.1990}.
The $\ell$th level of the hierarchy for graph isomorphism consists of
$n^{\Theta(\ell)}$ linear equalities in the variables
$\var{v_1\mapsto w_1,\ldots,v_\ell\mapsto w_\ell}$.  Now instead of
dropping the integrality constraints, as one typically does in
combinatorial optimisation, we
drop the nonnegativity constraints and are left with a system of linear Diophantine equations. 
Our \emph{Diophantine isomorphism test} solves this system over the
integers, which is possible in time $n^{\Theta(\ell)}$, and then
answers ``non-isomorphic'' if no solution exists.  What is remarkable
about this algorithm is that it distinguishes all pairs of CFI-graphs,
and not only that, but also the variants of the CFI-graphs modulo $p$
for all primes $p$. (The CFI-graphs may be viewed as graph encodings
of systems of linear equations modulo 2, and they have natural
variants modulo $p$.) As the CFI-graphs and their variants are used in
all previous lower bound proofs---arguably, the CFI-construction is
the only systematic construction of hard examples for graph
isomorphism that is known---this explains why no lower bounds for
the Diophantine isomorphism test
were known.

Algebraic algorithms for graph isomorphism start from similar
equations in variables $\var{v\mapsto w}$ as the integer linear
program, except that nonnegativity constraints are replaced by
polynomial equations $\var{v\mapsto w}^2=\var{v\mapsto w}$ to ensure
$\set{0,1}$-solutions.
These algorithms can best be analysed by algebraic proof systems such
as Nullstellensatz~\cite{Beame.1994} or the (more powerful) polynomial
calculus \cite{Clegg.1996}, which captures the power of the Gr\"obner
basis algorithm.  In this setting, non-isomorphic graphs can be
efficiently distinguished if they have a refutation of low degree over
some field $\mathbb F$.  In a previous paper
\cite{BerGro15}, 
we established 
degree
lower bounds for graph isomorphism in the polynomial calculus over all
fields except fields of characteristic~$2$. 
These lower bounds were obtained by a reduction from the
so-called Tseitin tautologies in a version due to Buss et al.~\cite{Buss.2001}, for which lower bounds were known in all characteristics but~$2$. 
In this paper, the lower bounds are based on a 
different construction due to Alekhnovich and Razborov~\cite{AR01LowerBounds}, which also provides hard instances over fields of characteristic $2$. 
More significantly, we construct families of pairs of non-isomorphic
graphs for which we can prove lower bounds for the polynomial calculus
that simultaneously hold for all prime fields. 

To prove the lower bounds, both for linear Diophantine equations and
the polynomial calculus over all prime fields, we cannot use the
CFI-instances, because they are distinguished by both algorithms (in
the case of polynomial calculus: the CFI-instances modulo $p$ are
distinguished over the field $\mathbb F_p$). In \cite{BerGro15}, we
established a close connection between the CFI-instances and what we
call \emph{group CSPs}, that is, constraint satisfaction problems
where the constraints are membership tests in some coset of a
cartesian power of a base group. For the ``classical'' CFI-instances
modulo $p$, this group is $\ZZ_p$.  We can associate a pair of
non-isomorphic graphs with any instance of
 an unsatisfiable
group CSP, but unfortunately, this generalisation still does not suffice for the lower bound proof.\footnote{We suspect that 
the Diophantine isomorphism test
can distinguish these graphs for all group CSPs, or at least all Abelian group CSPs, but we can only prove this for Abelian groups that are direct products of 
prime groups $\ZZ_p$.
}
A crucial new idea of this paper is to enhance the group CSPs by an
additional constraint of bounded size. This yields what we call an
\emph{$e$-extended group CSP}, where $e$ is the size (number of
permitted values) of the non-group constraint. We show that for every
fixed $e$, instances of $e$-extended group CSPs can still be
translated to pairs of non-isomorphic graphs. We apply this
construction to group CSPs over the group $\ZZ_2\times\ZZ_3$ and use
the non-group constraint to introduce a ``disjunction'' between the
subgroups $\ZZ_2\times\{0\}$ and $\{0\}\times\ZZ_3$. The resulting
pairs of non-isomorphic graphs are hard to distinguish for linear
Diophantine equations and the polynomial calculus simultaneously over
all prime fields.

We believe that our construction, while still rooted in the
CFI-constructions, adds a genuinely new aspect and thus provides new
hard graph-isomorphism instances which may be useful in other contexts
as well.

\subsection*{Related Work}
The connection between the linear programming approach to graph
isomorphism and the $1$-di\-mensional WL algorithm (a.k.a colour
refinement) goes back to Tinhofer~\cite{tin86}. The correspondence
between the levels of the Sherali-Adams hierarchy and the
higher-dimensional WL was established by Atserias and Maneva
\cite{atsman13} and independently Malkin \cite{mal14} and later
refined by Grohe and Otto~\cite{GroheOttoJournal.2016}. O'Donnell
et. al. \cite{DWWZ.2013} and Codenotti et al.~\cite{codschsno14}
proved that that even the more powerful semi-definite Lasserre
hierarchy fails to distinguish CFI-graphs.

A related approach of applying algebraic techniques to graph
isomorphism was initiated by the authors of this paper in
\cite{BerGro15}. We proved lower bounds for the polynomial calculus
over all fields of characteristic $\neq2$ and also established close
connections between the algebraic approach, the linear programming
approach, and the WL-algorithm. For a detailed discussion of
these
connections, we refer the reader to \cite{BerGro15}.

\section{Linear Equations for Graph Isomorphism and CSP}

\subsection{Preliminaries}

In general, we use standard notation and terminology, but let us
highlight a few points. We denote the vertex and edge set of a directed
or undirected graph $G$ by $V(G)$ and $E(G)$, respectively. We denote
the edges of an undirected graph by $vw$ (instead of $\{v,w\}$) and
the edges of a directed graph by $(v,w)$. An \emph{orientation} of an
undirected graph $G$ is a directed graph $D$ such that $V(D)=V(G)$
and $E(D)$ contains exactly one of $(v,w),(w,v)$ for all $vw\in
E(G)$. 

If $G$ is an undirected graph, for every set $W\subseteq V(G)$ we let $E(W)$ be the set of all edges
incident with a vertex in $W$ and 
$\Bd(W)$, the \emph{boundary} of $W$,  the set of all edges incident with a vertex in $W$ and a vertex
in $V\setminus W$. Note that $\Bd(W)=\Bd(V\setminus W)=E(W)\cap
E(V\setminus W)$. 
If $D$ is a directed graph, for every subset
$W\subseteq V(D)$, we let $\Bd_-(W)$ be the set of all
edges of $D$ with head in $W$ and tail in $V\setminus W$ and $\Bd_+(D)$ the set of all edges of $H$ with tail in $W$ and head in
$V\setminus W$. 
We write $\Bd(v)$  instead of $\Bd(\{v\})$, and similarly
$\Bd_-(v),\Bd_+(v)$.

Both for undirected and directed $G$ and $W\subseteq V(G)$, 
by $G[W]$ we denote the induced subgraph of $G$ with vertex set $W$,
and we let $G\setminus W:=G[V(G)\setminus W]$. Moreover, for
$F\subseteq E(G)$ we let $G-F:=(V(G),E(G)\setminus F)$.

The \emph{degree} of a vertex $v$ of an undirected graph is
$|\Bd(v)|$, and the \emph{degree} of a vertex $v$ of a directed graph
is $|\Bd_-(v)|+|\Bd_+(v)|$. A directed or undirected graph is
$d$-regular if every vertex has degree $d$.

Recall that an instance of the \emph{constraint satisfaction problem (CSP)} is a triple
$({X},\domD,\CC)$, where ${X}$ is a finite set of
\emph{variables}, $\domD$ a finite \emph{domain}, and $\CC$ a set of
constraints of the form $\big(\vec x,R\big)$, where
$\vec x\in{X}^k$ and $R\subseteq \domD^k$, for some $k\ge
0$. An \emph{assignment} $\phi:{X}\to\domD$ \emph{satisfies} the
constraint if $\phi(\vec x)\in R$. The
\emph{arity} of the constraint is $k$, and the \emph{arity} of the
instance $({X},\domD,\CC)$ is the maximum arity of its
constraints.
When the domain is clear from the context we specify CSPs by the set
$\CC$ of their constraints and let the variables be given
implicitly. In this case, we refer to the set of variables of $\CC$ by
$\Var(\CC)$ and to the domain by $\Dom(\CC)$.

\subsection{Equations for Graph Isomorphism and CSP}

Given two graphs $G, H$ we introduce for $\ell\ge 1$ a
system of linear equations $\Liso^\ell(G,H)$. 
These systems form a hierarchy $\Liso^1\subset\Liso^2\subset\cdots$
and are equivalent to the Sherali-Adams hierarchy of relaxations for
a natural linear programming formulation of the graph isomorphism problem, see \cite{BerGro15} for a more detailed discussion of encodings. 
The variables of $\Liso^\ell(G,H)$ are $\var\pi$ for sets $\pi\subseteq V(G)\times
V(H)$ of size $|\pi|\le \ell$. We interpret these sets as partial mappings from
$V(G)$ to $V(H)$. For sets $\pi$ that do not correspond to partial
mappings the system will have an equation $[\pi]=0$, and thus we can
ignore such $\pi$. To emphasise the partial-mapping view, we write
$\var{v_1\mapsto w_1,\ldots,v_m\mapsto w_m}$ or $\var{\vec
  v\mapsto\vec w}$ instead of
$\var{\{(v_1,w_1),\ldots,(v_m,w_m)\}}$. We also write $\var{\pi,v\mapsto
  w}$ instead of $\var{\pi\cup\{(v,w)\}}$.
We say that $\pi$ is a \emph{partial isomorphism}
from $G$ to $H$ if it is an injective partial mapping that additionally preserves adjacencies, that is
$vw\in E(G)\iff \pi(v)\pi(w)\in E(H)$. If $G$ and $H$ are coloured
graphs, partial isomorphisms are also required to preserve colours.
We let $\Liso^\ell(G,H)$ be the following system of
linear equations:
\begin{align}
\label{eq:Liso_1}
  \sum_{v\in V(G)}\var{\pi,v\mapsto w}&=\var \pi&\parbox[t]{7cm}{for all $\pi\subseteq V(G)\times
V(H)$ of size $|\pi|\le \ell-1$ and all $w\in V(H)$,}\\
\label{eq:Liso_2}
  \sum_{w\in V(H)}\var{\pi,v\mapsto w}&=\var \pi&\parbox[t]{7cm}{for all $\pi\subseteq V(G)\times
V(H)$ of size $|\pi|\le \ell-1$ and all $v\in V(G)$,}\\
\label{eq:Liso_3}
\var{\emptyset}&=1,\\
\label{eq:Liso_4}
\var{\pi}&=0&\parbox[t]{7cm}{for all $\pi\subseteq V(G)\times
V(H)$ of size $|\pi|\le \ell$ such that $\pi$ is not a partial isomorphism.}
\end{align}
Note that for all $\ell\geq 2$ the system $\Liso^\ell(G,H)$ has a nonnegative integral solution if and only if the graphs are isomorphic.

In the same way we define for a CSP $\CC$ and $\ell\ge 1$ the system
of linear equations $\Lcsp^\ell(\CC)$.
The variables of
our system are $\var \psi$ for sets $\psi\subseteq \Var(\CC)\times\Dom(\CC)$ of size $|\psi|\le \ell$. We interpret sets as partial mappings from
$\Var(\CC)$ to $\Dom(\CC)$, which are intended to be partial solutions, that is,
partial mappings that satisfy all constraints whose variables are in the
domain of $\psi$. We denote the domain of $\psi$ by $\dom(\psi)$ and
also use notations like $\var{x_1\mapsto \gamma_1,\ldots,x_m\mapsto \gamma_m}$ or $\var{\vec
  x\mapsto\vec \gamma}$  or $\var{\psi,x\mapsto
  \gamma}$.
We let $\Lcsp^\ell(\CC)$ be the following system of
linear equations:
\begin{align}
  \label{eq:lcsp1}
  \sum_{\gamma\in D}\var{\psi,x\mapsto \gamma}&=\var
                                                \psi&\parbox[t]{7cm}{for
                                                      all $\psi\subseteq
                                                      \Var(\CC)\times\Dom(\CC)$ of size $|\psi|\le \ell-1$ and all $x\in \Var(\CC)$,}\\
  \label{eq:lcsp2}
\var{\emptyset}&=1,\\
  \label{eq:lcsp3}
\var{\psi}&=0&\parbox[t]{7cm}{for all $\psi\subseteq \Var(\CC)\times\Dom(\CC)$ of
               size $|\psi|\le \ell$ such that $\psi$ is not a partial solution.}
\end{align}
If $\CC$ is a $k$-ary CSP-instance and $\ell\ge k$, then the system
$\Lcsp^\ell(\CC)$ has a nonnegative integral solution if and only if
$\CC$ is satisfiable.  We are interested in the (not necessarily
nonnegative) integral solutions of $\Liso$ and $\Lcsp$. As one step
towards this we also consider a certain type of rational solutions:
for an integer $p$, a \emph{\psolution} of a system of linear equations
is a satisfying assignment over $\QQ$ that only assigns values from
$\{0\}\cup\{p^z\mid z\in\ZZ\}$.
The next lemma states a criterion for the existence of integral
solutions.

\begin{lem}
  \label{lem:p-solution}
  Let $\mathsf L$ be a system of linear equations over $\ZZ$, and let
  $p,q\in\ZZ$ be co-prime. If $\mathsf L$ has a $p$-solution and a
  $q$-solution, then it has an integral solution.
\end{lem}

\begin{proof}
  Suppose that $\mathsf L$ is of the form
  $\setdescr{\vecx}{\matM\vecx=\vecb}$ for a matrix
  $\matM\in\ZZ^{\dimM\times\dimN}$ and a vector $\vecb \in
  \ZZ^\dimM$.
  Let $\vecx, \vecy \in \QQ^\dimN$ be solution vectors over $\{0\}\cup\{p^z\mid
  z\in\ZZ\}$ and $\{0\}\cup\{q^z\mid z\in\ZZ\}$, respectively.
  If one of these solutions is already integral there is nothing to prove.
  Otherwise, let $z\geq 1$ be maximal such that $\vecx$
  contains a value of the form $p^{-z}$ or $\vecy$
  contains a value of the form $q^{-z}$.
  Note that $p^{z}\vecx, q^{z}\vecy \in \ZZ^\dimN$.
  Because $p^{z}$ and $q^{z}$ are relatively prime, there are integers
  $\alpha, \beta \in \ZZ$ such that $\alpha p^{z}+\beta q^{z}=1$.
  Now we have
  \begin{equation}
    \label{eq:8}
    \matM\cdot(\alpha p^{z}\vecx+\beta q^{z}\vecy) = \alpha
    p^{z}\matM\vecx+\beta q^{z}\matM\vecy = \alpha
    p^{z}\vecb+\beta q^{z}\vecb = (\alpha
    p^{z}+\beta q^{z})\vecb = \vecb.
  \end{equation}
  Hence, $\alpha p^{z}\vecx+\beta q^{z}\vecy$ is an integral
  solution for $\mathsf L$.
\end{proof}

\subsection{Tseitin Tautologies}\label{sec:tseitin}
For a directed graph $H$, and Abelian group $\Gamma$, and a mapping
$\sigma\colon V(H)\to\Gamma$, the \emph{$\Gamma$-Tseitin tautology}
$\CC^{H,\Gamma,\sigma}$ is the following CSP with domain $\Gamma$ and variables $x_{e}$ for $e\in E(H)$.
For every $v\in V(H)$ of degree $\arityk$, the CSP
$\CC^{H,\Gamma,\sigma}$ has a $\arityk$-ary
constraint $C^{H,\Gamma,\sigma}(v)$ defined by the equation
\begin{equation}
  \sum_{e\in\Bd_+(v)}x_{e}-\sum_{e\in\Bd_-(v)}x_{e}=\sigma(v)  \label{eq:Tseitin_constraint}
\end{equation}
over $\Gamma$. 
It is easy to see
that if
$\sum_{v\in V(H)}\sigma(v)\neq 0$, then  $\CC^{H,\Gamma,\sigma}$ has no solution. 
For $\Gamma=\ZZ_2$, 
the $\Gamma$-Tseitin tautologies are the classical Tseitin tautologies
\cite{Tseitin1983}. 
The graph $H$ is typically a $k$-regular
expander graph (see Appendix~\ref{sec:expander}).

$\ZZ_p$-Tseitin tautologies were defined in \cite{Buss.2001} and have
been used to prove degree lower bounds for polynomial calculus over
all fields whose characteristic $q$ contains a primitive $p$th root of
unity.  In \cite{AR01LowerBounds}, this lower bound was extended to
all fields of characteristic $q\neq p$,
even for the more restricted variant
of \emph{Boolean $\ZZ_p$-Tseitin tautologies} $\CB^{H,\ZZ_p,\sigma}$, which have the same
variables and constraints as the $\ZZ_p$-Tseitin tautologies, but the
smaller domain $\{0,1\}$.  %

\section{Extended Group CSPs}
\label{sec:exgroupcsp}

We recall the notion of group CSPs introduced in \cite{BerGro15}, but restrict ourselves to finite Abelian groups $\Gamma$ (written additively).
An instance of a \emph{$\Gamma$-CSP} has domain $\Gamma$ and constraints of the form
$\big((x_1,\ldots,x_\arityk),\Delta+\bgamma\big)$, where $\Delta\le\Gamma^\arityk$ is a subgroup and $\Delta+\bgamma$ a coset for some
$\bgamma\in\Gamma^\arityk$. For example, the $\Gamma$-Tseitin
tautologies $\CC^{H,\Gamma,\sigma}$ are $\Gamma$-CSPs (but not their
Boolean versions).

With each constraint
$C=\big((x_1,\ldots,x_\arityk),\Delta+\bgamma\big)$, we associate the
\emph{homogeneous} constraint $\tilde
C=\big((x_1,\ldots,x_\arityk),\Delta\big)$. For an instance
$\CC$, we let $\tilde{\CC}=\{\tilde C\mid C\in\CC\}$.
For every group CSP $\CC$ we define two graphs $\graphcsp$, $\graphcsptilde$ that are isomorphic if and only if $\CC$ is satisfiable.
Let $\CC$ be a $\Gamma$-CSP. We construct a
coloured graph $G(\CC)$ as follows.
\begin{itemize}
\item For every variable $x\in\Var(\CC)$ we take vertices $\gamma^{(x)}$ for all
  $\gamma\in \Gamma$. 
  We colour all these vertices with
  a fresh colour $L^{(x)}$.
\item For every constraint $C=((x_1,\ldots,x_\arityk),\Delta+\bgamma)\in\CC$ we add
  vertices $\bbeta^{(C)}$ for all $\bbeta\in\Delta+\gamma$. 
  We colour
  all these vertices with a fresh colour $L^{(C)}$. If
  $\bbeta=(\beta_1,\ldots,\beta_\arityk)$, we add an edge
  $\{\bbeta^{(C)},\beta_i^{(x_i)}\}$ for all $i\in[\arityk]$. We
    colour this edge with colour $M^{(i)}$.
\end{itemize}
We let $\graphcsptilde$ be the graph $G(\tilde\CC)$ where for all
constraints $C\in\CC$ we identify the
two colours $L^{(C)}$ and $L^{(\tilde C)}$.
We call $\graphcsp$ and $\graphcsptilde$ the \emph{CFI-graphs} over $\CC$.
The
CFI-graphs over the $\ZZ_2$-Tseitin tautologies $\CC^{H,\ZZ_2,\sigma}$ are just the ``standard'' CFI-graphs,
going back to Cai, F\"urer, and Immerman~\cite{caifurimm92}. These graphs have
been intensively studied and applied in the finite-model-theory
literature (and elsewhere).

It was show in \cite[Lemma 2.1]{BerGro15} that the CFI-graphs over
$\CC$ are isomorphic if and only if $\CC$ is
satisfiable.
The next lemma additionally shows \psolutions can be transferred from
the corresponding $\Lcsp$ to $\Liso$.

\begin{lem}\label{lem:integral_solution_csp_to_gi}
  Let $\Gamma$ be an Abelian group and $\CC$ a $\Gamma$-CSP of arity $\arityk$. 
  \begin{enumerate}[label=(\alph*)]
  \item \label{item:csp_and_gi_isomorphic}
    $\CC$ is satisfiable if and only if $\graphcsp$ and $\graphcsptilde$ are isomorphic.
    \item \label{item:integral_solution_csp_to_gi} 
    If $\Lcsp^{\arityk\ell}(\CC)$ has a \psolution, then so does $\Liso^\ell(\graphcsp,\graphcsptilde)$. 
  \end{enumerate}
\end{lem}
\begin{proof}
  For the first statement we repeat the argument from \cite{BerGro15},
  showing that for every satisfying assignment $\phi$ of $\CC$ there
  is an isomorphism $\isopi_{\phi}$ between $\graphcsp$ and
  $\graphcsptilde$ and for every isomorphism $\isopi_{\phi}$ between
  $\graphcsp$ and $\graphcsptilde$ there is a satisfying assignment
  $\phi_\isopi$ of the $\CC$.  Let $G=(V,E):=G(\CC)$ and
  $\tilde G=(\tilde V,\tilde E):=\tilde G(\CC)$.  Let
  $\phi\colon\CX\to\Gamma$ be a satisfying assignment for $\CC$. We
  define a mapping $\isopi_\phi\colon V\to\tilde V$ as follows:
  \begin{itemize}
  \item For every $x\in\CX$ and $\gamma\in\Gamma$ we let
    $\isopi_\phi(\gamma^{(x)}):=\big(\gamma-\phi(x)\big)^{(x)}$.
  \item For every $C=(x_1,\ldots,x_\arityk,\Delta+\bgamma)\in\CC$ and every
    $\bbeta=(\beta_1,\ldots,\beta_\arityk)\in\Delta+\bgamma$ we let 
    \[
    \isopi_\phi(\bbeta^{(C)}):=\big(\beta_1-\phi(x_1),\ldots,\beta_\arityk-\phi(x_\arityk)\big)^{(C)}.
    \]
    To see that this is well defined, note that $\phi(\vec
    x):=\big(\phi(x_1),\ldots,\phi(x_\arityk)\big)\in\Delta+\bgamma$, because
    $\phi$ satisfies the constraint $C$. Thus 
    \[\bbeta - \phi(\vec
    x)=\big(\beta_1-\phi(x_1),\ldots,\beta_\arityk-\phi(x_\arityk)\big)\in\Delta.
    \] 
  \end{itemize}
  It is easy to see that the mapping $\isopi_\phi$ is bijective. To see that it
  is an isomorphism, consider, for some constraint
  $C=\bigl((x_1,\ldots,x_\arityk),\Delta+\bgamma\bigr)\in\CC$ and some $i\in[\arityk]$, a vertex
  $\bbeta^{(C)}$, where $\bbeta=(\beta_1,\ldots,\beta_\arityk)\in\Delta+\bgamma$, and a vertex
  $\gamma^{(x_i)}$, where $\gamma\in\Gamma$. Then 
  \begin{align*}
  \{\bbeta^{(C)},\gamma^{(x_i)}\}\in
  E&\iff\beta_i=\gamma\iff\beta_i-\phi(x_i)=\gamma-\phi(x_i)\\
  &\iff
  \{\isopi_\phi(\bbeta^{(C)}),\isopi_\phi(\gamma^{(x_i)})\}\in \tilde E.
  \end{align*}
  To prove the backward direction, suppose that $\isopi$ is an isomorphism
  from $G$ to $\tilde G$. We define an assignment $\phi_\isopi:\CX\to\Gamma$ by
  \[
  \phi_\isopi(x)^{(x)}=\isopi^{-1}(0^{(x)}).
  \]
  (Here $0^{(x)}$ denotes the $x$-copy of the
  unit element $0\in\Gamma$ in the graph $\tilde G$.) To see that
  $\phi_\isopi$ is a satisfying assignment, consider a constraint
  $C=(x_1,\ldots,x_\arityk,\Delta+\bgamma)\in\CC$. Let
  $\bbeta=(\beta_1,\ldots,\beta_\arityk)$ with $\beta_i=\phi_\isopi(x_i)$. 
  We need to prove that $\bbeta\in\Delta+\bgamma$. We have
  $\isopi(\beta_i^{(x_i)})=0^{(x_i)}$. As $\vec 0=(0,\ldots,0)\in\Delta$,
  the vertex $\vec 0^{(\tilde C)}\in \tilde V$ has edges to all
  vertices $\isopi(\beta_i^{(x_i))})$. Thus the vertex $\isopi^{-1}(\vec
  0^{(\tilde C)})$ has colour $L^{(C)}=L^{(\tilde C)}$ and edges to
  the vertices $\beta_i^{(x_i)}$. This implies that $\isopi^{-1}(\vec
  0^{(\tilde C)})=\balpha^{(C)}$ for some $\balpha\in\Delta+\bgamma$ and
  $\balpha=(\beta_1,\ldots,\beta_\arityk)=\bbeta$.
This concludes the proof of \ref{item:csp_and_gi_isomorphic}.
\\
For the translation \ref{item:integral_solution_csp_to_gi} of satisfying assignments from $\Lcsp^{\arityk\ell}$ to $\Liso^\ell$  let $\Phicsp\colon \Var\bigl(\Lcsp^\ell(\CC)\bigr)\to \QQ$ be a \psolution. 
We define an assignment $\Phiiso\colon \Var\bigl(\Liso^\ell(\graphcsp,\graphcsptilde)\bigr)\to \QQ$ as follows. We let $\Phiiso(\var{\emptyset})=1$.
Let $\isopi=\{(v_1, w_1),\ldots,(v_\indm, w_\indm)\}$. 
If $\isopi$ is not a partial isomorphism we let $\Phiiso(\var{\isopi})=0$. 
Otherwise, all $(v_i, w_i)$ are of the form $(\gamma^{(x)},\gamma'^{(x)})$ or $(\balpha^{(C)},\bbeta^{(C)})$.
We let 
\begin{align}
  \mappsi_{\isopi} \defeq \qquad &\{(x,\gamma - \gamma')\mid (\gamma^{(x)},\gamma'^{(x)}) \in \isopi\} \\
  \cup \;&\{(x_j,\alpha_j - \beta_j)\mid
           ((\alpha_1,\ldots,\alpha_\aritykalt)^{(C)},(\beta_1,\ldots,\beta_\aritykalt)^{(C)})
           \in \isopi\\
  &\hspace{2.6cm}\text{for some
           }C=((x_1,\ldots,x_{k'}),\Delta+\bgamma)\in\CC\text{ and
           }1\le j\le k'\}
\end{align}
 and set $\Phiiso(\var{\isopi}) \defeq \Phicsp(\var{\mappsi_\isopi})$, noting that $\setsize{\mappsi_\isopi} \leq \arityk\setsize{\isopi}\leq k\ell$. 
 It is clear that if $\Phicsp$ takes values from $\{0\}\cup\{p^z\mid z\in \ZZ\}$ for some prime $p$, then so does $\Phiiso$.
We have to check that this assignment satisfies the equations \eqref{eq:Liso_1}--\eqref{eq:Liso_4} from $\Piso^\ell$.  
First note that \eqref{eq:Liso_3} and \eqref{eq:Liso_4} are satisfied by definition. 
We show that $\Phiiso$ satisfies all equations of the form \eqref{eq:Liso_1} $\sum_{w\in V(H)}\var{\isopi,v\mapsto w}=\var \isopi$ for some $v$, the argument for \eqref{eq:Liso_2} is symmetric. 
First suppose that $v=\gamma^{(x)}$, it follows that
\begin{align}
   \sum_{w\in V(H)}\Phiiso(\varnormal{\isopi,\gamma^{(x)}\mapsto w}) 
   &= \sum_{\gamma'\in \Gamma}\Phiiso(\varnormal{\isopi,\gamma^{(x)}\mapsto \gamma'^{(x)}}) \\
   &= \sum_{\gamma'\in \Gamma}\Phicsp(\varnormal{\mappsi_\isopi,x\mapsto \gamma'}) \label{eq:locallabel1}
   = \Phicsp(\varnormal{\mappsi_\isopi})
   = \Phiiso(\varnormal{\isopi}),
 \end{align} 
where \eqref{eq:locallabel1} follows from \eqref{eq:lcsp1}.
Now suppose that $v=\bbeta^{(C)}$ for some $\bbeta=(\beta_1,\ldots,\beta_\aritykalt)$ and $C=((x_1,\ldots,x_\aritykalt), \Delta+\bgamma)$.
Similar as above we have
\begin{align}
  \sum_{w\in V(H)}\Phiiso(\varnormal{\isopi,\bbeta^{(C)}\mapsto w}) 
  &= \sum_{\balpha\in \Gamma^\aritykalt}\Phiiso(\varnormal{\isopi,\bbeta^{(C)}\mapsto \balpha^{(C)}}) \\
  &= \sum_{\balpha\in \Gamma^\aritykalt}\Phicsp(\varnormal{\mappsi_\isopi,x_1\mapsto \alpha_1,\ldots,x_\aritykalt\mapsto \alpha_\aritykalt}) \label{eq:locallabel3}\\
  &= \Phicsp(\varnormal{\mappsi_\isopi}) \label{eq:locallabel4}
  = \Phiiso(\varnormal{\isopi}).
\end{align}
This concludes the proof of \ref{item:integral_solution_csp_to_gi}.
\end{proof}

Now we extend group CSPs by small non-group constraints
and provide a similar graph encoding for them as for group CSPs.
Extended group CSPs will later play a crucial role in the lower bound
arguments for the Diophantine equations
(Section~\ref{sec:diophantine}) and the polynomial calculus %
(Section~\ref{sec:PClowerbounds}).  Let $\Gamma$ be
a finite Abelian group.  An $\extensionparameter$-\emph{extended}
$\Gamma$-CSP has a constraint set $\CCext=\CC\cup\{\Carb\}$, where
$\CC$ defines a $\Gamma$-CSP, and $\Carb = (\vec x,\Rarb)$ is an
additional constraint with $\setsize{\Rarb}\leq \extensionparameter$.
For $\bgamma\in \Rarb$ we let
$\CC_\bgamma := \CC \cup \{(\vec x,\{\bgamma\})\}$ be the CSP obtained
by fixing the variables in the constraint $\Carb$. Observe that
$\CC_\bgamma$ is a $\Gamma$-CSP, for every $\bgamma\in\Gamma^k$, even
if $\CC^*$ is not.
Furthermore, $\CC\cup\{\Carb\}$ is satisfiable if and only if there
exists an $\bgamma\in \Rarb$ such that $\CC_\bgamma$ is satisfiable.
We now show how to encode extended group CSPs into instances of graphs
isomorphism, that is, we prove an analogue of
Lemma~\ref{lem:integral_solution_csp_to_gi} for extended group CSPs. Let us start by
reviewing a well-known ``or-construction'' for graph
isomorphism (see \cite{kobtorsch93}). For mutually disjoint graphs
$G_1,\ldots,G_\ell$, we let $G_1\uplus\ldots\uplus G_\ell$ be the
disjoint union of the $G_i$ (we also write
$\biguplus_{i=1}^\ell G_i$), and we let
$\langle G_1,\ldots,G_\ell\rangle$ be the graph obtained from the
disjoint union $G_1\uplus\ldots\uplus G_\ell$ by adding fresh vertices
$v_1,\ldots,v_\ell$ and edges from $v_1$ to all vertices in $V(G_1)$
and from $v_i$ to all vertices in $V(G_{i-1})\cup V(G_i)$ for all
$i\ge 2$.
Thus $\langle G_1,\ldots,G_\ell\rangle$ encodes an ordered sequence of the graphs $G_i$ 
and it is not hard to show that two sequence graphs $\langle G_1,\ldots,G_\ell\rangle$ and $\langle H_1,\ldots,H_\ell\rangle$ are isomorphic if and only if all pairs $G_i, H_i$ are isomorphic.

\begin{defn}\label{def:or_construction}
  Let $(G^0_i,G^1_i)_{i\in [\ell]}$ be a sequence of pairs of graphs. 
  We define the graph pair $(G^0,G^1) = \bigvee_{i\in [\ell]}(G^0_{i},G^1_i)$ as follows. For $j\in\{0,1\}$ let
  \begin{align}
    G^j &= \biguplus\, \Bigl\{\langle G^{a_1}_{1},\ldots,G^{a_\ell}_{\ell}\rangle \Bigmid \textstyle\sum_{i=1}^\ell a_i \equiv j \pmod 2\Bigr\}.
  \end{align}
\end{defn}

\begin{lem}\label{lem:or_construction}
 Let $(G^0,G^1) = \bigvee_{i\in [\ell]}(G^0_{i},G^1_i)$. Then $G^0$ and $G^1$ are isomorphic if and only if there exists an $i$ such that $G^0_{i},G^1_i$ are isomorphic.
\end{lem}

\begin{proof}
  Suppose that $G^0$ and $G^1$ are isomorphic. 
  As both graphs consist of $2^{\ell-1}$ connected components, every isomorphism is a combination of isomorphisms between the sequence graphs and hence between the components of the corresponding sequence graphs. 
  As all pairs of sequence graphs from $G^0$ and $G^1$ differ in at least one component, it follows that some pair $G^0_{i}$, $G^1_i$ has to be isomorphic. 
  For the other direction suppose that  $G^0_i$ and $G^1_i$ are isomorphic. 
  There is a bijection between the sequence graphs of $G^0$ and $G^1$ such that $\langle G^{a_1}_{1},\ldots,G^{a_\ell}_{\ell}\rangle$ is matches with $\langle G^{b_1}_{1},\ldots,G^{b_\ell}_{\ell}\rangle$ where $b_i = 1-a_i$ and $b_j=a_j$ for $j\neq i$.
  By combining the isomorphism between  $G^0_i$ and $G^1_i$ with automorphisms on $G^{a_j}_j$ for $j\neq i$ it follows that all pairs of sequence graphs and hence  $G^0$ and $G^1$ are isomorphic. 
\end{proof}

\begin{lem}\label{lem:extended_group_to_gi}
  Suppose that  $\CCext = \CC\cup (\vec x, \Rarb)$ is an $\extensionparameter$-extended group CSP of arity $\arityk$ and let $(G^0_\CCext,G^1_\CCext) \defeq \bigvee_{\bgamma\in \Rarb}(G(\CC_{\bgamma}),\tilde G(\CC_{\bgamma}))$.
\begin{enumerate}[label=(\alph*)]
 \item \label{item:extended_group_to_gi-a}
 $G^0_\CCext$ and $G^1_\CCext$ are isomorphic if and only if $\CCext$ is satisfiable. 
  \item \label{item:extended_group_to_gi-b}
  The size of $G^0_\CCext$ and $G^1_\CCext$ is bounded by $\Bigoh{2^\extensionparameter\setsize{\CCext}}$ 
    \item \label{item:extended_group_to_gi-d}
  If  $\Lcsp^{\arityk\ell}(\CC_\bgamma)$ has a \psolution for some $\bgamma\in\Rarb$, then  $\Liso^\ell(G^0_\CCext,G^1_\CCext)$ has a \psolution. 
\end{enumerate}
\end{lem}

\begin{proof}
  For $\CCext=\CC\cup\{(\vec x,\Rarb)\}$ consider the $\Gamma$-CSPs
  $\CC_{\bgamma}$ as defined above and let $G^0_{{\bgamma}}\defeq
  G(\CC_{\bgamma})$ and $G^1_{{\bgamma}}\defeq \tilde
  G(\CC_{\bgamma})$ be the corresponding CFI-graphs, which are of size
  $\bigoh{\setsize{\CC_\bgamma}}$.
  By definition, $G^0_\CCext$ and $G^1_\CCext$ have size $\Bigoh{2^\extensionparameter\|\CCext\|}$ and by Lemma~\ref{lem:or_construction} they are isomorphic if and only if $G^0_{\CC_{\bgamma}}$ and $G^1_{\CC_{\bgamma}}$ are isomorphic for some $\bgamma\in\Gamma$. 
  As this holds if and only if the corresponding $\CC_{\bgamma}$ is satisfiable, it follows that $G^0_\CCext$ and $G^1_\CCext$ are isomorphic if and only if $\CCext$ is satisfiable.

For \ref{item:extended_group_to_gi-d} suppose that $\Lcsp^{\arityk\ell}(\CC_\bgamma)$ has a \psolution.
From Lemma~\ref{lem:integral_solution_csp_to_gi}\ref{item:integral_solution_csp_to_gi} it follows that there is a \psolution for $\Liso^\ell(G^0_{\CC_\bgamma},G^1_{\CC_\bgamma})$ and $\Liso^\ell(G^1_{\CC_\bgamma},G^0_{\CC_\bgamma})$. 
We can fix a bijection between the sequence graphs in $G^0_\CCext$ and $G^1_\CCext$ such that every pair of sequence graphs differs only in component $\bgamma$. 
To define the \psolution, we first set $\var{\pi}=0$ if $\pi$ is not a
partial isomorphism between the corresponding components of the
sequence graphs that are matched by the bijection. 
Otherwise, let $\pi$ be a mapping between components $G^0_i$ and $G^1_i$ of
two matched sequence graph.
If $G^0_i=G^1_i$, then we set $\var{\pi}=1$ if it is a subset of the
identity mapping and $\var{\pi}=0$, else.
If $G^0_i\neq G^1_i$, one of both graphs is a copy of $G^0_{\CC_\bgamma}$ and the
other is a copy of $G^1_{\CC_\bgamma}$.
In this case we let $\var{\pi}$ as defined by the corresponding \psolution for
$\Liso^\ell(G^0_{\CC_\bgamma},G^1_{\CC_\bgamma})$ or
$\Liso^\ell(G^1_{\CC_\bgamma},G^0_{\CC_\bgamma})$.
\end{proof}

\section{Lower Bounds for Linear Diophantine Equations}
\label{sec:diophantine}

In this section we prove one of our main results, the lower bound for
the \emph{Linear-Diophantine-Equations algorithm} for graph
isomorphism testing. Recall that the algorithm works as follows. The
input consists of two graphs $G,G'$, in addition we have a parameter
$\ell\ge 1$. The algorithm computes the system $\Liso^\ell(G,G')$ of
linear equations with integer coefficient (see
\eqref{eq:Liso_1}--\eqref{eq:Liso_4}) and solves it over the integers
(for example by using the polynomial time algorithm described in
\cite{schri86}). If the system has no solution, the the algorithm
answers ``not isomorphic''. If it has a solution, the algorithm
answers ``possibly isomorphic''. The running time of the algorithm is
$n^{O(\ell)}$, where $n$ is the number of vertices of the input
graphs. The algorithm is clearly sound, that is, always gives a
correct answer. To show that it is complete (for some $\ell$), we
would have to prove that for all pairs $G,G'$ of non-isomorphic input
graphs the system $\Liso^\ell(G,G')$ has no integral solution. Our theorem
shows that this fails in rather strong sense: for $\ell=o(n)$, there
are non-isomorphic input graphs for which the system does have an
integral solution.

\begin{theo} \label{thm:diophantine_lowerbound}
  For every $\ell\ge 1$ there are non-isomorphic 3-regular graphs $G,\tilde
  G$ of size $|G|=|\tilde G|=O(\ell)$ such that $\Liso^\ell(G,\tilde G)$ has an integral solution.
\end{theo}

The rest of this section is devoted to a proof of this theorem.
Let $\CE$ be a family of 2-connected 3-regular expander graphs. For the
necessary definitions and the existence of such a family we refer the
reader to Appendix~\ref{sec:expander}.
 The only consequence of the expansion
property that we use is stated in the following lemma, which is proved in the appendix (as
  Corollary~\ref{cor:exp}).

\begin{lem}\label{lem:exp}
  There is  constant $c>0$ such that for every $G\in\CE$ and
  every set $X\subseteq E(G)$ there is a set
  ${\hat X}\supseteq X$ of size $|{\hat X}|\le c|X|$ such that $E(G)\setminus
  {\hat X}$ is either empty or the edge set of a 2-connected subgraph of $G$.
\end{lem}

For the rest of this section, we fix a graph $G\in\CE$ and let $V:=V(G)$, $E:=E(G)$,
$n:=|V|$, and $m:=|E|$. Note that $m=(3/2)n$, because $G$ is
3-regular. We let 
\[
\ell:=\floor{\frac{m-1}{3c}}.
\]
For every $X\subseteq E$, let $K_X=(W_X,Z_X)$ be the subgraph of $G$
with edge set $Z_X:=E\setminus {\hat X}$ (with ${\hat X}$ from
Lemma~\ref{lem:exp}) and vertex set $W_X$ consisting of all vertices
incident with an edge in $Z_X$. Then
$K_X$ is either empty or 2-connected. If
$|X|\le\ell$, then 
\[
|Z_X|\ge m-c\cdot\ell>\frac{2}{3}m
\]
and thus 
\begin{equation}
  \label{eq:5}
  |W_X|>\frac{2}{3}n,
\end{equation}
because a graph of maximum degree $3$ with more than $(2/3)m$ edges
has more than $(2/3)(2/3)m=(2/3)n$ vertices. 

We say that a set $X\subseteq E$ is \emph{closed} if every edge $e\in
E\setminus X$ is contained in a cycle $Z\subseteq E\setminus
X$. Here, for simplicity, we identify a cycle with its edge set. Note
that the intersection of two closed sets is closed. Hence for every
set $X\subseteq E$ the set
\[
\cl(X):=\bigcap_{\substack{Y\supseteq X\\Y\text{ closed}}}Y,
\]
which we call the \emph{closure} of $X$, is closed, and in fact the unique
inclusionwise minimal closed set that contains $X$.  Observe that
$x\not\in\cl(X)$ if and only if there is a cycle $Z\in E\setminus X$
such that $x\in Z$. The forward direction of this equivalence is
immediate from the definition of closed sets, and for the backward
direction, note that if $Z$ is a cycle then $Y=E\setminus Z$ is a closed set.

We note that the
operator $\cl:2^E\to 2^E$ is a closure operator in the sense of
matroid theory (see \cite[Section~1.4]{oxl11}), that is,
\begin{itemize}
\item $X\subseteq\cl(X)=\cl(\cl(X))$ for all $X\subseteq E$,
\item $X\subseteq Y$ implies $\cl(X)\subseteq\cl(Y)$ for all
  $X,Y\subseteq E$,
\item $y\in\cl(X\cup\{x\})\setminus\cl(X)$
implies $x\in\cl(X\cup\{y\})$ for all $X\subseteq E$ and $x,y\in
E$.
\end{itemize}
To verify the last property, known as the \emph{exchange
  property}, suppose that $y\in\cl(X\cup\{x\})\setminus\cl(X)$ and
$x\not\in\cl(X\cup\{y\})$. Then there is a cycle $Z_x\subseteq
E\setminus (X\cup\{y\})$ with $x\in Z_x$. As $y\not\in\cl(X)$, there
is a cycle $Z_y\subseteq E\setminus X$ with $y\in Z_y$. As
$y\in\cl(X\cup\{x\})$, we have $x\in Z_y$. But then $(Z_y\cup
Z_x)\setminus\{x\}$ contains a
cycle through $y$. This cycle has an empty intersection with
$X\cup\{x\}$, which contradicts $y\in\cl(X\cup\{x\})$.

Let us denote the matroid by $\CM$. The independent sets of $\CM$ are
the sets $I\subseteq E$
such that $x\not\in\cl(I\setminus \{x\})$ for all $x\in I$ (see \cite[Theorem~1.4.4]{oxl11}). A
\emph{basis} for $\CM$ is an inclusionwise maximal independent set,
and a \emph{basis} for a set $X\subseteq E$ is an inclusionwise maximal
independent set $I\subseteq X$. All bases of a set $X$ have the same
cardinality, the \emph{rank} $\rk(X)$. The \emph{rank} of the matroid
$\CM$ is $\rk(E)$.

Now let $Z$ be the edge set of a 2-connected subgraph of $G$. Then
$E\setminus Z$ is closed. Thus for every set $X$ we have
$\cl(X)\subseteq {\hat X}$ for the set ${\hat X}$ of Lemma~\ref{lem:exp}, and it
follows from the lemma that
\[
|\cl(X)|\le c\cdot|X|.
\]
This implies that the rank of the matroid $\CM$ is at least
$\ceil{m/c}$, because if $B$ is a basis of $\CM$ then $\cl(B)=E$.

We let $H$ be an arbitrary orientation of $G$ and $\vE:=E(H)$.  We let
$\Gamma$ be a finite Abelian group, $\sigma:V\to\Gamma$. We consider
the Tseitin tautology $\CC:=\CC^{H,\Gamma,\sigma}$ (see
Section~\ref{sec:tseitin}). Recall that the set of variables of $\CC$
is
\[
\Var(\CC)=\{x_e\mid e\in\vE\},
\]
and for each $v\in V$ the CSP $\CC$ has a constraint $C(v)$ expressed
by the following equation in the group $\Gamma$:
\begin{equation}
  \label{eq:2}
  \sum_{e\in \Bd_+(v)}x_e-\sum_{e\in \Bd_-(v)}x_e=\sigma(v).
\end{equation}
It will be convenient for
us to think of an undirected edge $vw\in E$, its orientation $(v,w)$ or
$(w,v)\in\vE$, and the variable $x_{(v,w)}$ or $x_{(w,v)}$ as the same
object, that is, identify the sets $E$ and $\vE$ and $\Var(\CC)$. Generically, we denote the
set by $E$, subsets by $X,Y,Z$, and elements by $x,y,z$, but
sometimes, we still denote the elements by
$vw$, $(v,w)$, or $x_{(v,w)}$ to indicate which role of an element we are thinking of at the moment.

For every subset $W\subseteq V$ we let $\sigma(W):=\sum_{w\in
  W}\sigma(w)$, and we let $C(W)$ be the constraint
\begin{equation}
  \label{eq:4}
  \sum_{e\in \Bd_+(W)}x_e-\sum_{e\in \Bd_-(W)}x_e=\sigma(W).
\end{equation}
The constraints $C(W)$ are not contained in $\CC$, but they are
implied by the constraints $C(v)$
of $\CC$, because equation~\eqref{eq:4} is just the sum of the
equations \eqref{eq:2} for $v\in W$. Thus every solution to $\CC$
satisfies all constraints $C(W)$. However, it is not the case that
every partial solution $\psi$ satisfies all constraints $C(W)$ with
$\Bd(W)\subseteq\dom(\psi)$, despite the fact that $\Bd(W)$ is the set of
all variables appearing in the constraint $C(W)$. 

For $k\ge 0$, we call $\psi\subseteq\Var(\CC)\times\Gamma$
\emph{$k$-consistent} if it is a partial mapping and for all
$W\subseteq V$ of size $|W|\le k$, if $\Bd(W)\subseteq\dom(\psi)$ then
$\psi$ satisfies the constraint $C(W)$. Note that $\psi$ is a partial
solution if and only if it is $1$-consistent. %

\begin{lem}\label{lem:dio1}
  Let $X\subseteq E$ and $\psi:X\to\Gamma$. Then
  $\psi$ is $k$-consistent if and only if the constraint $C(W)$ is
  satisfied for every $W\subseteq V$ such that $|W|\le k$ and $W$ is
  the vertex set of a connected component of the graph
  $G-X=(V,E\setminus X)$.
\end{lem}

\begin{proof}
  Let $W_0\subseteq V$ such that $|W_0|\le k$ and $\partial(W_0)\subseteq X$. Note that
  for every connected component $W$ of $G-X$, if
  $W\cap W_0\neq\emptyset$ then $W\subseteq W_0$, because otherwise
  there are $w\in W\cap W_0$ and $w'\in W\setminus W_0$ such that
  $ww'\in E\setminus X$, which contradicts the assumption
  $\partial(W_0)\subseteq X$. Thus $W$ is the union of vertex sets
  $W_1,\ldots,W_p$ of connected components of $G\setminus X$. But then
  $\partial_+(W)=\bigcup_{i=1}^p\partial_+(W_i)$ and
  $\partial_-(W)=\bigcup_{i=1}^p\partial_-(W_i)$, and the equation
  \eqref{eq:4} is just the sum of the corresponding equations for the
  $W_i$, which are satisfied by the assumption of the lemma.
\end{proof}

\begin{lem}\label{lem:dio2}
  Let $X\subseteq E$ such that $\rk(X)\le\ell$, and let
  $\psi:X\to\Gamma$ be $n/3$-consistent. Then $\psi$ is $2n/3$-consistent. 
\end{lem}

\begin{proof}
  Let $Y$ be a basis for $X$. Choose ${\hat Y}$ according to
  Lemma~\ref{lem:exp} and note that $Y\subseteq
  X\subseteq\cl(Y)\subseteq {\hat Y}$.
  Let $W_Y$ be the vertex set of the 2-connected graph $K_Y$ with edge
  set $Z_Y=E\setminus {\hat Y}$. Then $|W_Y|>2n/3$ by \eqref{eq:5}.
  Suppose for contradiction that $X$ is not $2n/3$-consistent. Then
  there is a set $W\subseteq V$ such that $|W|\le 2n/3$ and
  $\Bd(W)\subseteq X$ and $\psi$ does not satisfy $C(W)$. By
  Lemma~\ref{lem:dio1}, we may assume that $W$ is the vertex set of a
  connected component of $G-X$. As $\psi$ is
  $n/3$-consistent, we have $|W|>n/3$. Hence $W\cap
  W_Y\neq\emptyset$. As $K_Y$ is a connected subgraph of $G-X$, it follows that $W_Y\subseteq W$. Hence $|W|>2n/3$, which is
  a contradiction.
\end{proof}

We call $\psi$ \emph{robustly consistent} if it is $n/3$-consistent.

\begin{lem}\label{lem:cl-ext}
  Let $X\subseteq E$ such that $\rk(X)\le\ell$, and let
  $\psi:X\to\Gamma$ be robustly consistent. Then
  there is a unique 
  $\hat\psi:\cl(X)\to\Gamma$ such that $\psi\subseteq\hat\psi$ and
  $\hat\psi$ is robustly consistent.
\end{lem}

\begin{proof}
  It suffices to prove that we can uniquely extend $\psi$ to a domain
  $X\cup\{x\}$ for an $x\in\cl(X)\setminus X$. 
  So let us pick such an $x$. 
    Let $(W,Z)$ be the connected component of $G-X$ that contains
    $x$. As $x$ is not contained in a cycle in $G-X$ (otherwise it would
    not be in $\cl(X)$), the edge
    $x$ is a bridge (that is, separating edge) of the graph
    $(W,Z)$. Let $(W_1,Z_1)$ and $(W_2,Z_2)$ be the two connected
    components of $(W,Z-\{x\})$. Without loss of generality we
    assume that $|W_1|\le |W_2|$. Then $|W_1|\le n/2$, and in order to
    be robustly satisfiable, the mapping
    $\psi'$ we shall define must satisfy the constraint $C(W_1)$.
    As $\Bd(W_1)\subseteq X\cup\{x\}$ and the values
    $\psi'(x')=\psi(x')$ are fixed for all $x'\in X$, there is a unique
    $\gamma\in\Gamma$ such that $\psi':=\psi\cup\{(x,\gamma)\}$
    satisfies the constraint $C(W_1)$. 

    If $|W_2|\le n/3$, then
    $|W|=|W_1|+|W_2|\le 2n/3$, and as $\psi$ is robustly consistent,
    it satisfies the constraint $C(W)$. This implies that $\psi'$
    satisfies the constraint $C(W_2)$.

   All other connected components $(W',Z')$ of $G-(X\cup\{x\})$ are also
   connected components of $G-X$. Thus $\psi'$ satisfies the
    constraint $C(W')$ if and only $\psi$ does, and this implies that
    $\psi'$ is robustly consistent.
\end{proof}

\begin{lem}\label{lem:free-ext}
  Let $X\subseteq E$ such that $\rk(X)\le\ell-1$, and let
  $\psi:X\to\Gamma$ be robustly consistent. Let $x\in
  E\setminus\cl(X)$. Then every  
  $\psi':X\cup\{x\}\to\Gamma$ such that $\psi\subseteq\psi'$ is robustly consistent.
\end{lem}

\begin{proof}
  By Lemma~\ref{lem:cl-ext}, we may assume without loss of generality
  that $X$ is closed.  Let $\gamma\in\Gamma$ and
  $\psi':=\psi\cup\{(x,\gamma)\}$. Let $(W,Z)$ be a connected
  component of the graph $G-X$ that contains $x$. Then $(W,Z-\{x\})$
  is connected, because $x\not\in X=\cl(X)$. Note that $\psi'$ satisfies
  $C(W)$ because $\psi$ does and $\Bd(W)\subseteq X$.

  All other connected components $(W',Z')$ of $G-(X\cup\{x\})$ are also connected components of $G-X$. Thus $\psi'$
  satisfies the constraint $C(W')$ if and only $\psi$ does, and this
  implies that $\psi'$ is robustly consistent.
\end{proof}

\begin{cor}
  Let $Y$ be an independent set of the matroid $\CM$ of size
  $|Y|\le\ell$. Then every mapping $\psi:Y\to\Gamma$ is robustly consistent. 
\end{cor}

We are now ready to turn to the linear program
$\LL:=\Lcsp^\ell(\CC)$. Recall that the variables of $\LL$ are $\var\psi$
for $\psi\subseteq \Var(\CC)\times\Gamma$ of size at most $\ell$. 
We define an assignments $\Psi:\Var(\LL)\to \mathbb Q$ by
\begin{equation}
\label{eq:7}
\Psi(\var\psi):=
\begin{cases}
  \displaystyle\frac{1}{|\Gamma|^r}&\text{if $\psi$ is robustly consistent and
    $\rk(\dom(\psi))=r$},\\
  0&\text{if $\psi$ is not robustly consistent}.
\end{cases}
\end{equation}

\begin{lem}
  $\Psi$ is a solution to $\LL$.
\end{lem}

\begin{proof}
  As all robustly consistent $\psi$ are partial solutions, $\Psi$
  satisfies the equations \eqref{eq:lcsp3}. 

  The empty mapping is robustly consistent, because $G$ is connected
  and thus the only component of $(V,E\setminus\dom(\emptyset))=G$
  contains more than $(1/3)n$ vertices. As $\rk(\emptyset)=0$, we have
  $\Psi(\var\emptyset)=1$, and thus $\Psi$ satisfies \eqref{eq:lcsp2}. 
  
  To see that $\Psi$ satisfies the equations \eqref{eq:lcsp1}, let
  $\psi\subseteq\Var(\CC)\times\Gamma$ of size $|\psi|\le\ell-1$ and $x\in
  E$. We have to prove that
  \begin{equation}
    \label{eq:6}
    \sum_{\gamma\in\Gamma}\Psi(\var{\psi\cup\{(x,\gamma)\}})=\Psi(\var\psi).
  \end{equation}
Let $X:=\dom(\psi)$ and $r:=\rk(X)$. If $\psi$ is not robustly
  consistent then neither is $\psi\cup\{(x,\gamma)\}$ for any $\gamma$, and
  both sides of equation \eqref{eq:6} are zero. Suppose that $\psi$ is
  robustly consistent. Then
  \[
  \Psi(\var\psi)=\frac{1}{|\Gamma|^r}.
  \]
  Suppose first that $x\in\cl(X)$. Then by Lemma~\ref{lem:cl-ext} there is a unique
  $\gamma_x\in\Gamma$ such that $\psi':=\psi\cup\{(x,\gamma_x)\}$ is robustly
  consistent. As $x\in\cl(X\cup\{x\})$, we have $\rk(X\cup\{x\})=r$ and
  thus $\Psi(\var{\psi'})=\frac{1}{|\Gamma|^r}$. Hence
  \[
  \sum_{\gamma\in\Gamma}\Psi(\var{\psi\cup\{(x,\gamma)\}})=\Psi(\var{\psi'})=\frac{1}{|\Gamma|^r}=\Psi(\var\psi).
  \]
Suppose next that $x\not\in\cl(X)$. Then by Lemma~\ref{lem:free-ext},
for all $\gamma\in\Gamma$ the mapping
$\psi\cup\{(x,\gamma)\}$ is robustly consistent. Moreover,
$\rk(X\cup\{x\})=r+1$ and thus $\Psi(\var{\psi\cup\{(x,\gamma)\}})=\frac{1}{|\Gamma|^{r+1}}$. Hence
  \[
  \sum_{\gamma\in\Gamma}\Psi(\var{\psi\cup\{(x,\gamma)\}})=|\Gamma|\cdot \frac{1}{|\Gamma|^{r+1}}=\frac{1}{|\Gamma|^r}=\Psi(\var\psi).
  \]
  Thus $\Psi$ satisfies \eqref{eq:6} and hence all equations \eqref{eq:lcsp1}.
\end{proof}

Recall that \emph{$p$-solution} for
a system of linear equations is a rational solution that only takes
values $p^z$ for integers $z$. A $p$-group is a
group of order $p^k$ for a nonnegative integer $k$.

\begin{cor}\label{cor:p-group}
  Let $p$ be a prime and $\Delta\le\Gamma$ be a $p$-group. Suppose that
  $\sigma(v)\in\Delta$ for all $v\in V(H)$. Then
  $\Lcsp^\ell(\CC^{H,\Gamma,\sigma})$ has a $p$-solution.
\end{cor}

\begin{proof}
  We note that every solution $\Psi$ to
  $\Lcsp^\ell(\CC^{H,\Delta,\sigma})$ can be extended to a solution
  $\Psi'$ to $\Lcsp^\ell(\CC)$ by letting
  $\Psi'(\var\psi):=\Psi(\var\psi)$ for all
  $\psi\subseteq\Var(\CC)\times\Delta$ and $\Psi'(\var\psi):=0$ for all
  $\psi\not\subseteq\Var(\CC)\times\Delta$. Thus we can apply the
  previous lemma to the group $\Delta$ and the CSP $\CC^{H,\Delta,\sigma}$.
\end{proof}

\begin{proof}[Proof of Theorem~\ref{thm:diophantine_lowerbound}] We
  let $\Gamma$ be the group $\ZZ_2\times\ZZ_3$. Let $\Delta_2$ be the
  subgroup $\ZZ_2\times\{0\}$ and $\Delta_3$ the subgroup
  $\{0\}\times\ZZ_3$. Moreover, let $\iota_2:=(1,0)$ and
  $\iota_3:=(0,1)$. We continue to work with the same graph $G$ and
  orientation $H$ of $G$ as before. We choose an arbitrary $v^*\in
  V$.
  We let $\sigma_2,\sigma_3:V\to\Gamma$ by
  $\sigma_2(v):=\sigma_3(v):=(0,0)$ for $v\in V\setminus\{v^*\}$ and
  $\sigma_p(v^*):=\iota_p$.

  For $p=2,3$, we let $\Psi_p$ be the $p$-solution to
  $\Lcsp^\ell(\CC^{H,\Gamma,\sigma_p})$ obtained from
  Corollary~\ref{cor:p-group} applied to the $p$-group
  $\Delta=\Delta_p$ and $\sigma=\sigma_p$.

  We now build a 2-extended $\Gamma$-CSP $\CC^*$ that is essentially
  the disjunction between $\CC^{H,\Gamma,\sigma_2}$ and
  $\CC^{H,\Gamma,\sigma_3}$. Note that $\CC^{H,\Gamma,\sigma_2}$ and
  $\CC^{H,\Gamma,\sigma_3}$ both have the constraints \eqref{eq:2}
  with $\sigma(v)=(0,0)$ for all $v\in V\setminus\{v^*\}$; they only
  differ in the constraints for $v^*$. To define $\CC^*$ we add a new
  variable $x^*$ and replace the constraints \eqref{eq:2} for $v^*$ by
  \[
  \sum_{e\in\partial_+(v)}x_{e}-\sum_{e\in\partial_-(v)}x_{e}=x^*,
  \]
  which still is a $\Gamma$-constraint. Now we add the unary
  constraint $\big(x^*,\{\iota_2,\iota_3\}\big)$, which is not a group
  constraint. $\CC^*$ is the resulting 2-extended $\Gamma$-CSP. For
  $p=2,3$, we let $\CC^*_p:=\CC^*_{\iota_p}=\CC^*\cup\{(x^*,\{\iota_p\})$.
  Furthermore, we let $\Psi_p^*:\Var(\Lcsp^\ell(\CC^*))\to\QQ$ be the assignment
  defined by 
  \[
  \Psi^*_p(\var\psi):=
  \begin{cases}
    0&\text{if $\psi$ is not a partial mapping},\\
    0&\text{if $p=2$ and $(x^*,\iota_3)\in\psi$ or $p=3$ and
      $(x^*,\iota_2)\in\psi$},\\
    \Psi_p\big([\var{\psi\setminus\{(x^*,\iota_p)\}}\big)&\text{otherwise}.
  \end{cases}
  \]
  It is easy to see that $\Psi^*_p$ is a $p$-solution to $\Lcsp^\ell(\CC^*_p)$.
  
  Note that all constraints of $\CC$ are ternary, because the graph
  $G$ is 3-regular. Let $\ell':=\floor{\ell/3}$ and
  \[
  (G,\tilde G):=\bigvee_{p\in\{2,3\}}\big(G(\CC^*_p),\tilde
  G(\CC^*_p)\big).
  \]
  By
  Lemma~\ref{lem:extended_group_to_gi}\ref{item:extended_group_to_gi-a},
  the graphs $G$ and $\tilde G$ are non-isomorphic, and by
  Lemma~\ref{lem:extended_group_to_gi}\ref{item:extended_group_to_gi-d},
  the system $\Liso^{\ell'}(G,\tilde G)$ has a $p$-solution. Thus by
  Lemma~\ref{lem:p-solution}, the system has an integral solution.
\end{proof}

\section{Polynomial Calculus}
\label{sec:PCdefinitions}

We now turn to an algebraic approach and encode instances of the
isomorphism problem by systems of polynomial equations, which we may
interpret over any field. Then we try to derive the non-solvability
of the system using algebraic reasoning. Again, we obtain an
algorithm that is sound, but not necessarily
complete. The algorithm is parameterized by the degree $\ell$ of the polynomials that
we see during the derivation, and its running time
$n^{O(\ell)}$.  We
shall prove a lower bound by exhibiting non-isomorphic graphs that
requiring degree $\ell=\Omega(n)$. The proper framework for phrasing
these results is propositional proof complexity.

\emph{Polynomial Calculus (PC)} \cite{Clegg.1996} is a proof system to
prove that a given system of (multivariate) polynomial equations $\SP$
over a field $\mathbb F$ has no $\set{0,1}$-solution.  We always normalise
polynomial equations to the form $p=0$ and just write $p$ to denote
the equation $p=0$.  Consequently, we view $\SP$ as a set of
polynomials.
Polynomials are derived line by line according to
the following derivation rules (for polynomials $p \in\SP$,
polynomials $f, g$, variables $x$ and field elements $a,b$):
$$
\frac{}{p}, \quad\frac{}{x^2-x}, \quad\frac{f}{xf}, \quad\frac{g\quad f}{ag+bf}.
$$
The \emph{axioms} of the systems are all $p\in\SP$ and
$x^2-x$ for all variables $x$.
A PC refutation of $\SP$ is a derivation of $1$ (the contradiction $1=0$).
The \emph{degree} of a PC derivation is the maximal degree of every polynomial in the derivation. 
If an instance $\SP$ is unsatisfiable and has a refutation of degree $d$, then it can be found in  time $n^{O(d)}$ by a bounded degree variant of the Gr\"obner basis algorithm \cite{Clegg.1996}.
To solve a combinatorial problem by this algebraic approach, one first
encodes the instance into a set of low degree polynomials $\SP$ and
then tries to find a PC refutation of degree $\degd$
over some field $\mathbb F$.
If such a refutation is found, we know that the instance is
unsatisfiable and the algorithm rejects.
Otherwise, the algorithm outputs ``possibly satisfiable''.
As for the Diophantine isomorphism test and other related approaches
such as linear and semi-definite programming hierarchies this
algorithm is sound but not necessarily complete.
It can be shown, however, that completeness is achieved for $d=n+1$
over any field (where $n$
is the total number of variables in $\SP$).

For polynomials $f_1,\ldots,f_\ell$ and $g$ over $\mathbb F$ we write
$\{f_1,\ldots,f_\ell\}\models_{\mathbb F} g$ if $g$ follows
semantically from $f_1,\ldots,f_\ell$, that is, for every
$\set{0,1}$-assignment $\assignI$ it holds that $\assignI(f_1)=0$, \ldots,
$\assignI(f_\ell)=0$ implies $\assignI(g)=0$.  By
$\{f_1,\ldots,f_\ell\}\vdash_{\mathbb F} g$ we denote that there is a
PC derivation of $g$ from $f_1,\ldots,f_\ell$ over $\mathbb F$ and use
$\{f_1,\ldots,f_\ell\}\vdash^d_{\mathbb F} g$ if there is a refutation
of degree at most $d$.  For prime fields $\mathbb F_p$ we abbreviate
$\models_{\mathbb F_p}$, $\vdash_{\mathbb F_p}$,
$\vdash^d_{\mathbb F_p}$ by $\models_{p}$, $\vdash_{p}$,
$\vdash^d_{p}$. The following theorems will be useful for us.
\begin{theo}[Derivational completeness (Theorem 5.2 in \cite{BussIPRS97})]\label{thm:derivational_completeness}
  Let $f_1,\ldots,f_\ell$ and $g$ be polynomials in $n$ variables and $p$ a prime. Then 
\begin{equation}
\{f_1,\ldots,f_\ell\}\models_p g     
\quad \Longleftrightarrow\quad \{f_1,\ldots,f_\ell\}\vdash_p g
\quad \Longleftrightarrow\quad \{f_1,\ldots,f_\ell\}\vdash^{n+1}_p g.
\end{equation}
\end{theo}
\begin{theo}[Cut-elimination (Theorem 5.1 (2) in \cite{BussIPRS97})]\label{thm:cut_elimination}
  Let $F:=\{f_i(\vec x,y)\}$  be a set of polynomials $f_i(\vec x,y)$ in variables $x_1,\ldots,x_\ell,y$ and $p$ a prime number. 
  Let  $F_0:=\{f_i(\vec x,0)\}$, $F_1:=\{f_i(\vec x,1)\}$ and $g$ be a polynomial. Then
  $$
  F_0 \vdash^{d}_p g \text{ and } F_1 \vdash^{d}_p g \Longrightarrow  F \vdash^{d+1}_p g.
  $$
\end{theo}

To compare the power of the polynomial calculus for different systems
of polynomials, we use \emph{low degree reductions} \cite{Buss.2001}.
Fix a field $\mathbb F$ and let $\SP$ and $\SQ$
be two sets of polynomials in the variables
$\CX$ and $\CY$, respectively. A
\emph{degree-$(d_1,d_2)$ reduction} from $\SP$ to $\SQ$ 
is a set of degree-$d_1$ polynomials $\{f_y\mid y\in\CY\}$ in the variables  $\CX$ such that
\begin{align}
&\SP\vdash^{d_2}_{\mathbb F} q(f_{y_1},\ldots,f_{y_\ell})  
&&\text{ for all }q(y_1,\ldots,y_\ell)\in\SQ\text{ and } \\
&\SP\vdash^{d_2}_{\mathbb F}f_y^2-f_y 
&&\text{ for all }y\in\CY.
\end{align}

\begin{lem}[Lemma~1 in \cite{Buss.2001}]\label{lem:lowdeg}
If there is a degree-$(\refdegd_1,\refdegd_2)$ reduction from $\SP$ to $\SQ$ and $\SQ$ has a PC refutation of degree $\refdegd$, then $\SP$ has a PC refutation of degree \mbox{$\max(\refdegd_2,\refdegd\cdot\refdegd_1)$}. 
\end{lem}

\subsection{Polynomial Encodings for Isomorphism and Constraint Satisfaction}

Following \cite{BerGro15}, for graphs $G,H$, we define a system of
(multivariate) polynomials $\Piso(G,H)$ in variables $\givar{v}{w}, v\in V(G),w\in V(H)$. 
A $\{0,1\}$-solution to the system is
intended to describe an isomorphism $\iota$ from $G$ to $H$, where
$\givar{v}{w}\mapsto 1$ if $\iota(v)=w$ and $\givar{v}{w}\mapsto0$ otherwise.
The system $\Piso(G,H)$ consists of the following linear and quadratic
polynomials:
\begin{align}
  &-1 + \textstyle\sum_{v\in V(G)} \givar{v}{w}  & &\text{for all }w \in V(H) \label{eq:cont1}\\
  &-1 + \textstyle\sum_{w\in V(H)} \givar{v}{w}  & &\text{for all }v \in V(G) \label{eq:cont2} \\
&\givar{v}{w}\cdot\givar{v'}{w'}  & &\parbox[t]{6cm}{for all $v,v'\in V(G),w,w'\in
                        V(H)$ such that $\{(v,w),(v',w')\}$ is no partial isomorphism.}\label{eq:comp}
\end{align}
Similarly, for every CSP $\CC$, we define a system of
(multivariate) polynomials $\Pcsp(\CC)$ in variables
$\cspvar{x}{\gamma}$ for $x\in\Var(\CC),\gamma\in\Dom(\CC)$.
A $\{0,1\}$-solution to the system is
intended to describe an solution $\alpha$, where
$\cspvar{x}{\gamma}\mapsto 1$ if $\alpha(x)=\gamma$ and
$\cspvar{x}{\gamma}\mapsto 0$ otherwise. 
The system $\Pcsp(\CC)$ consists of the following linear and quadratic polynomials:
\begin{align}
  &-1+\textstyle\sum_{\gamma\in\domD}\cspvar{x_i}{\gamma} &&\text{for all }x_i,\\
 &\cspvar{x_i}{\gamma}\cdot \cspvar{x_i}{\gamma'} &&\text{for all }x_i\text{ and }\gamma \neq \gamma',\\
  &\textstyle\prod^k_{i=1}\cspvar{x_i}{\gamma_i}&&\text{for all
                                                     constraints
                                                     $((x_1,\ldots,x_k),R)\in\CC$}\\
\notag
&&&\text{and all }(\gamma_1,\ldots,\gamma_\ell)\notin \relR.
\end{align}
Again,  $\Pcsp(\CC)$ has a
$\{0,1\}$-solution over some field $\mathbb F$ if and only if $\CC$ is
satisfiable.
Note the similarities between the  polynomial systems $\Piso(G,H)$ and
  $\Pcsp(\CC)$ and
  the linear systems $\Liso^\ell(G,H)$ and
  $\Lcsp^\ell(\CC)$.
  One formal correspondence is the following.
  Suppose that \LisoLcsp  viewed as a system of linear
  congruencies modulo some prime $p$ has no solution, then
  \PisoPcsp has a degree $\ell$ refutation over
  $\mathbb F_p$.
  Thus, the algebraic approach is stronger than solving the
  linear equations modulo some prime $p$, in that it is able to reject
  more unsatisfiable instances in time $n^{\bigoh{\ell}}$.
  On the other hand, solving the linear system over the integers is
  also more powerful than solving the system modulo some prime $p$ (because there
  might be a solution over $\ZZ_p$ even though the system has no
  solution over $\ZZ$).
  In fact, the Diophantine isomorphism test that solves \LisoLcsp
  over $\ZZ$ is incomparable in its strength with the algebraic
  approach of finding a polynomial calculus refutation of
  degree $\ell$.

\section{Lower Bounds for Polynomial Calculus}
\label{sec:PClowerbounds}

In this section we prove the following lower bound, which implies that
there are non-isomorphic graphs 
that cannot be distinguished in polynomial time by algebraic reasoning
over any prime field.

\begin{theo} \label{thm:PC_lowerbound}
  For every $\ell\ge 1$ there are non-isomorphic graphs $G,\tilde
  G$ of size $|G|=|\tilde G|=O(\ell)$ such that every polynomial
  calculus refutation of $\Piso(G,\tilde G)$ over some prime field
  $\mathbb F_p$ has degree $\Omega(\ell)$.
\end{theo}

One main ingredient in our proof is the 
framework of Alekhnovich and Razborov \cite{AR01LowerBounds} for
proving degree lower bounds.  They consider Boolean CSPs defined over
an expander graph, where the variables correspond to the edges of the
graph and where every constraint is defined over the edges that are
incident to the same vertex $v$.  To show that such CSPs are hard to
refute in polynomial calculus, they introduced the concept of
\emph{immunity}.  A constraint $C$ has high immunity over a field
$\mathbb F$, if it has no non-trivial low degree consequences, that
is, if from $\mathsf P_C\models_{\mathbb F} \polyg$ it follows
that either $\polyg\equiv 1$ or the degree of $\polyg$ is large
(linear in the number of variables in $C$).  The main result of
\cite{AR01LowerBounds} is that if $\CC$ is a Boolean CSP defined over
an expander graph and every constraint $C\in\CC$ is immune over
$\mathbb F$, then $\Pcsp(\CC)$ requires polynomial calculus
refutations of linear degree over $\mathbb F$.
As one example consider the Boolean $\ZZ_p$-Tseitin tautologies (see
Section~\ref{sec:tseitin}) where
for every vertex $v$ there is a constraint stating that the sum of
ingoing minus outgoing edge variables is congruent to $0$ modulo $p$.
Such parity constraints were shown to have high immunity over fields
of characteristic $\neq p$ \cite{AR01LowerBounds}, but they have low
immunity over $\mathbb F_p$ because the constraint
\eqref{eq:Tseitin_constraint} viewed as a linear equation over
$\mathbb F_p$ follows semantically over $\mathbb F_p$.  By the
immunity argument this implies the following lower bound.

\begin{theo}[Corollary 4.6 in \cite{AR01LowerBounds}]\label{thm:tseitin_mod_p}
  For every prime $p$ there is a constant $\arityk_0(p)$ such that the following holds.  
  For every $\arityk$ there is a directed $\arityk$-regular graph
  $\dirH$ on $n$ vertices such that for all $p$ with
  $\arityk_0(p)\leq\arityk$ it holds that every polynomial calculus
  refutation of
  $\Pcsp\bigl(\CB^{H,\ZZ_p,\sigma}\bigr)$
  over a field $\mathbb F$ of characteristic $\neq p$ requires degree $\Omega(\arityk n)$.
\end{theo}

To obtain lower bounds that hold over any field we cannot apply this
framework directly to $\Gamma$-CSPs, as for every Abelian group
$\Gamma$ the constraints always have low immunity over some prime
field $\mathbb F_p$.
Because of this we have to use extended group CSPs and show that there
is a  2-extended $(\mathbb Z_2\times\mathbb Z_3)$-CSP $\CC$ that is at least
as hard as the Boolean $\ZZ_p$-Tseitin tautologies for
$p\in\set{2,3}$.
This implies that $\Pcsp(\CC)$ 
requires large polynomial calculus degree over every prime field and the same
holds true for its graph encoding via the or-construction over the CFI-graphs. 

In Section~\ref{sec:exgroupcsp} we have already shown that solvability
and $p$-solutions can be transferred from ($e$-extended) group CSPs to
the graph isomorphism problem.  We now prove the corresponding
statements in the algebraic setting, showing that low degree
refutations for the graph isomorphism imply low degree refutations for
the system of polynomials corresponding to the underlying CSP.

\begin{lem}\label{lem:reduction_csp_gi}
  Let $\mathbb F_p$ be a prime field, $\Gamma$ an Abelian group, $\CC$ an $n$-variable $\Gamma$-CSP of arity $\arityk$ and $\refdegd_0 = \bigl(\arityk\setsize{\Gamma}+\setsize{\Gamma}^\arityk\bigr)^2+1$. 
  \begin{enumerate}[label=(\alph*)]
    \item \label{item:low_degree_reduction_csp_to_gi}
    There is a degree-$(\arityk, \refdegd_0)$ reduction from $\Pcsp(\CC)$ to $\Piso(\graphcsp,\graphcsptilde)$.
    \item \label{item:low_degree_reduction_gi_to_csp}
    There is a degree-$(1, \refdegd_0)$ reduction from $\Piso(\graphcsp,\graphcsptilde)$ to  $\Pcsp(\CC)$.
  \end{enumerate}
\end{lem}

\begin{proof}
We define the a degree-$(\arityk, \refdegd_0)$ reduction from $\Pcsp(\CC)$ to $\Piso(\graphcsp,\graphcsptilde)$ as follows. 
\begin{align}
\polyf_{\givar{\gamma^{(x)}}{(\gamma-\alpha)^{(x)}}} &\defeq \cspvar{x}{\alpha} \\
\polyf_{\givar{(\beta_1,\ldots,\beta_\aritykalt)^{(C)}}{(\beta_1-\alpha_1,\ldots,\beta_\aritykalt-\alpha_\aritykalt)^{(C)}}} &\defeq \prod^\aritykalt_{i=1}\cspvar{x_i}{\alpha_i} \text{ for }C=((x_1,\ldots,x_\aritykalt),\Delta+\bbeta)\\
\polyf_{\givar{v}{w}} &\defeq 0 \text{ for all other }v,w 
\end{align}
It is easy to see that all polynomials $\polyf^2_{\givar{v}{w}}-\polyf_{\givar{v}{w}}$ are derivable.
Let $\polyg$ be one of the substituted axioms
$\sum_{v\in V(G)} \polyf_{\givar{v}{w}} - 1$, $  \sum_{w\in V(H)} \polyf_{\givar{v}{w}} - 1$, or $\polyf_{\givar{v}{w}}\polyf_{\givar{v'}{w'}}$.
Note that there is a constraint $C\in\CC$ such that all variables in $\polyg$ are of the form $\cspvar{x}{\gamma}$ for some variable $x$ occurring in $C$.
Let $\mathsf P_C\subseteq \Pcsp(\CC)$ be
the set of polynomials for constraint $C$. 
By Theorem~\ref{thm:derivational_completeness} it suffices to show that $\mathsf P_C \models_p \polyg$ as this implies that there is a degree $\arityk\setsize{\Gamma}+1\leq \refdegd_0$ derivation of $\polyg$ from $\Pcsp(\CC)$.
Suppose that there is a $\set{0,1}$-assignment $\assignI$ that satisfies $\mathsf P_C$, we have to show that $\assignI(\polyg)=1$. 
By the definition of $\Pcsp$ it follows that there is a satisfying assignment $\phi$ for the constraint such that $\phi(x)=\gamma$ if and only if $\assignI(\cspvar{x}{\gamma})=1$. 
Let $\isopi_\phi$ be the isomorphism between the corresponding
subgraphs as defined in the proof of Lemma~\ref{lem:integral_solution_csp_to_gi} and note that the definition of the substitution gives $\assignI(\polyf_{\givar{v}{w}}) = 1$ if and only if $\isopi_\phi(v)=w$. 
Hence, $\assignI(\polyg)=1$ as every axiom from  
$\Piso(\graphcsp,\graphcsptilde)$
is satisfied by a $\set{0,1}$-assignment corresponding to an isomorphism.

For the backward direction \ref{item:low_degree_reduction_gi_to_csp} we define the degree-$(1, \refdegd_0)$ reduction from $\Piso(\graphcsp,\graphcsptilde)$ to $\Pcsp(\CC)$ by
   $\polyf_{\cspvar{x}{\gamma}} \defeq \givar{0^{(x)}}{(-\gamma)^{(x)}}$. 
   The argument that this is indeed a low degree reduction is similar to the one for \ref{item:low_degree_reduction_csp_to_gi}. 
Let $\polyg$ be some substituted axiom, $C\in\CC$ the corresponding constraint, and $\mathsf P_C\subseteq \Piso(\graphcsp,\graphcsptilde)$ be
the set of polynomials over two subgraphs of $\graphcsp$ and $\graphcsptilde$ that encode the constraint $C$. 
Note that both subgraphs have at most
$\arityk\setsize{\Gamma}+\setsize{\Gamma}^\arityk$ vertices and therefore
$\mathsf P_C$ contains at most
$\bigl(\arityk\setsize{\Gamma}+\setsize{\Gamma}^\arityk\bigr)^2 =
\refdegd_0-1$ variables.
By Theorem~\ref{thm:derivational_completeness} it now suffices to show that $\mathsf P_C \models_p \polyg$.
Suppose that there is a $\set{0,1}$-assignment $\assignI$ that satisfies $\mathsf P_C$. 
By the definition of $\Piso$ it follows that there is an isomorphism $\isopi$ between the gadgets such that $\isopi(0^{(x)})=(-\gamma)^{(x)}$ if and only if $\assignI\bigl(\givar{0^{(x)}}{(-\gamma)^{(x)}}\bigr)=1$. 
Let $\phi_\isopi$ be the corresponding satisfying assignment for $C$
(from Lemma~\ref{lem:integral_solution_csp_to_gi})  and note that the definition of the substitution gives $\assignI(\polyf_{\cspvar{x}{\gamma}}) = 1$ if and only if $\phi_\isopi(x)=\gamma$. 
Hence, $\assignI(\polyg)=1$ as every axiom from  
$\Pcsp(\CC)$
is satisfied by a $\set{0,1}$-assignment corresponding to a satisfying assignment for the CSP.
\end{proof}

The next Lemma transfers Lemma~\ref{lem:extended_group_to_gi} to the
algebraic setting and  provides a reduction from
$\extensionparameter$-extended group CSPs to graph isomorphism.

\begin{lem}\label{lem:extended_group_to_gi_PC}
  Fix a prime field $\mathbb F_p$.
  Suppose that $\CCext = \CC\cup (\vec x, \Rarb)$ is an
  $\extensionparameter$-extended group CSP of arity $\arityk$ and let
  $(G^0_\CCext,G^1_\CCext) \defeq \bigvee_{\bgamma\in
    \Rarb}(G(\CC_{\bgamma}),\tilde G(\CC_{\bgamma}))$.  If
  $\Piso(G^0_\CCext,G^1_\CCext)$ has a degree-$\refdegd$ refutation,
  then $\Pext$ has a refutation of degree
  $\bigoh{\refdegd\arityk\setsize{\Gamma}}$.
\end{lem}

\begin{proof}
  For $\bgamma\in \Rarb$ consider the $\Gamma$-CSPs $\CC_{\bgamma}$
  and let $G^0_{{\bgamma}}\defeq G(\CC_{\bgamma})$ and
  $G^1_{{\bgamma}}\defeq \tilde G(\CC_{\bgamma})$ be the corresponding
  CFI-graphs.
We first apply cut-elimination to $\Pext$ and the at most $\arityk\setsize{\Gamma}$ variables $\cspvar{x_i}{\gamma}$ where $x_i$ is a variable occurring in the additional constraint $(\vec x,\Rarb)$.
  By Theorem~\ref{thm:cut_elimination} it follows that
  \begin{equation}
    \label{eq:proofstmt1}
    \Pcsp(\CC_{\bgamma})\vdash^{\refdegd\arityk}_{p} 1 \quad \Longrightarrow \quad \Pext\vdash^{\refdegd\arityk+\arityk\setsize{\Gamma}}_{p}1.
  \end{equation}
  By Lemma~\ref{lem:reduction_csp_gi} there are degree-$(\arityk, \refdegd_0)$ reductions from $\Pcsp(\CC_{\bgamma})$ to $\Piso(G^0_\bgamma,G^1_\bgamma)$ for some constant $\refdegd_0$.
  By Lemma~\ref{lem:lowdeg} this implies for every $\bgamma\in\Rarb$ and sufficiently large $\refdegd\geq\refdegd_0$
\begin{equation}
  \label{eq:proofstmt2}
  \Piso(G^0_\bgamma,G^1_\bgamma)\vdash^\refdegd_{p} 1 \quad \Longrightarrow \quad \Pcsp(\CC_{\bgamma})\vdash^{\refdegd\arityk}_{p} 1
\end{equation}
Finally we show that there is a degree-$(1,2)$ reduction from $\Piso(G^0_\bgamma,G^1_\bgamma)$ to $\Piso(G^0_\CCext,G^1_\CCext)$.
As mentioned in the proof of Lemma~\ref{lem:or_construction} there is a bijection $\pi$ between the sequence graphs $\widehat G$ contained in $G^0_\CCext$ and $G^1_\CCext$ such that every pair of sequence graphs differs only in component $\bgamma$. 
By fixing all other components we reduce the isomorphism test for $G^0_\bgamma$ and $G^1_\bgamma$ to testing isomorphism of one component. 
We denote vertices in $G^0_\CCext$ and $G^1_\CCext$ by $(v, \widehat G, G^j_\balpha)$ referring to the vertex $v$ in the corresponding copy of $G^j_\balpha$ that is contained in the sequence $\widehat G$.
\begin{align}
  f_{\givar{(v,\widehat G, G^j_\bgamma)}{(w,\pi(\widehat G), G^{1-j}_\bgamma)}} &\defeq \givar{v}{w} & &\text{for $v\in V(G^0_\bgamma)$, $w\in V(G^1_\bgamma)$, $j\in\{0,1\}$,} \\
  f_{\givar{(v,\widehat G, G^j_\balpha)}{(v,\pi(\widehat G), G^j_\balpha)}} &\defeq 1 & &\text{for $v\in V(G^0_\balpha)$, $\balpha\in\Rarb\setminus\{\bgamma\}$, $j\in\{0,1\}$,} 
\end{align}
  and $f_{\givar{\mathsf v}{\mathsf w}} \defeq 0$ in all other cases. 
 As this reduction turns every axiom of $\Piso(G^0_\CCext,G^1_\CCext)$ into a trivial polynomial or an axiom of 
$\Piso(G^0_\bgamma,G^1_\bgamma)$ they can be derived immediately in degree $2$. 
By Lemma~\ref{lem:lowdeg} it follows that
\begin{equation}
  \label{eq:proofstmt3}
  \Piso(G^0_\CCext,G^1_\CCext)\vdash^\refdegd_{p} 1
  \quad\Longrightarrow\quad 
  \Piso(G^0_\bgamma,G^1_\bgamma)\vdash^\refdegd_{p} 1 \text{
    for all } \bgamma\in\Rarb. 
\end{equation}
The lemma follows by combining \eqref{eq:proofstmt1}, \eqref{eq:proofstmt2}, and \eqref{eq:proofstmt3}. 
\end{proof}

Now we have everything in hand to prove our lower bound.

\begin{proof}[Proof of Theorem~\ref{thm:PC_lowerbound}]
  Let be $\dirH$ a $\arityk$-regular directed graph such that
  $\arityk\geq\max(\arityk_0(2),\arityk_0(3))$ satisfies the
  conditions of Theorem~\ref{thm:tseitin_mod_p} for $p=2$ and $p=3$.
  We choose an arbitrary vertex $v^\ast\in V(\dirH)$, let $\sigma(v^\ast):= 1$ and $\sigma(v):= 0$ for all $v\in
  V(\dirH)\setminus\set{v^\ast}$, and consider the Boolean Tseitin
  CSPs $\CB^{H,\ZZ_2,\sigma}$ and $\CB^{H,\ZZ_3,\sigma}$ in the 
  variable set $\setdescr{x_e}{e\in E(H)}$.
    
  We define
  the unsatisfiable 2-extended $\Gamma$-CSP $\CCext$
  for $\Gamma=\mathbb Z_2\times \mathbb Z_3$ in variables
  $\setdescr{y_e}{e \in E(\dirH)}\cup\set{y^\ast}$ as in the proof of
  Theorem~\ref{thm:diophantine_lowerbound}.
  That is, we let $\CCext$ be the Tseitin tautology
  $\CC^{H,\Gamma,\sigma^\ast}$ for $\sigma^\ast\equiv (0,0)$ where we
  replace 
  the constraints \eqref{eq:2} for $v^*$ by the $\Gamma$-constraint 
  \[
  \sum_{e\in\partial_+(v)}y_{e}-\sum_{e\in\partial_-(v)}y_{e}=y^*
  \]
  and add the unary non-group
  constraint $\big(y^*,\{\iota_2,\iota_3\}\big)$ for $\iota_2:=(1,0)$
  and $\iota_2:=(0,1)$.
  Intuitively, $\CC^\ast$ is the Tseitin tautology
  $\CC^{H,\Gamma,\sigma^\ast}$ where we have $\sigma^\ast(v) = (0,0)$ for
  all $v\in V(H)\setminus\set{v^\ast}$ and the additional constraint
  that either $\sigma^\ast(v^\ast) = (0,1)$ or  $\sigma^\ast(v^\ast) = (1,0)$.
We
construct simple low degree reductions from $\Pcsp\bigl(\CB^{H,\ZZ_2,\sigma}\bigr)$ as well
as from $\Pcsp\bigl(\CB^{H,\ZZ_3,\sigma}\bigr)$ to $\Pext$, which are
in fact just restrictions. Fix $p\in\{2,3\}$. For the
reduction from $\Pcsp\bigl(\CB^{H,\ZZ_p,\sigma}\bigr)$ we set for all
$(u,v)\in E(\dirH)$
\begin{align}
  f_{\cspvar{y_{(u,v)}}{(0,0)}} &:= \cspvar{x_{(u,v)}}{0}, \\ 
  f_{\cspvar{y_{(u,v)}}{\iota_p}} &:= \cspvar{x_{(u,v)}}{1},  \\
  f_{\cspvar{y_{(u,v)}}{\gamma}} &:= 0, \quad\text{ if }\gamma \notin \{(0,0),\iota_p\}, 
\end{align}
Furthermore, for the additional variable $y^\ast$ we set
\begin{align} 
  f_{\cspvar{y^\ast}{\iota_p}} &:= 1,  \\
  f_{\cspvar{y^\ast}{\gamma}} &:= 0, \quad\text{ if }\gamma \neq \iota_p, 
\end{align} 
We have to check that this substitution fulfils the requirements of
low degree reductions.  As every variable $y$ from
$\Pext$ is substituted by a polynomial $f_{y}$ of the form $0$,
$1$, $\cspvar{x_e}{0}$, or
$\cspvar{x_e}{1}$, the equations
$f_y^2-f_y$ follow immediately.
Furthermore, the additional constraint
$\big(y^*,\{\iota_2,\iota_3\}\big)$ is satisfied.
In order to verify
that all substituted axioms from vertex constraints
$C_v$ have a constant degree derivation from $\Pcsp\bigl(\CB^{H,\ZZ_p,\sigma}\bigr)$, we apply Theorem~\ref{thm:derivational_completeness} and note
that each substituted constraint
$C_v$ follows semantically from the corresponding vertex constraint
\eqref{eq:Tseitin_constraint} in
$\Pcsp\bigl(\CB^{H,\ZZ_p,\sigma}\bigr)$.
As every such constraint involves at most
$\arityk$ variables we know that the substituted equations can be
derived in degree $\arityk+1$.
As both low degree reductions hold over every prime field, it follows by Lemma~\ref{lem:lowdeg}
and Theorem~\ref{thm:tseitin_mod_p} that every polynomial calculus
refutation of $\Pext$ over a prime field requires degree
$\Omega(\setsize{\Pext})$.
Because $\Pext$ is a 2-extended group CSP, the lower bound for $\Piso(G,\tilde G)$ follows from Lemma~\ref{lem:extended_group_to_gi_PC}.
\end{proof}

%
%
%
%


  \appendix
  \section{Expanders}
\label{sec:expander}

We just review the bare essentials of expander graphs and refer the
reader to the survey~\cite{hoolinwig06} for background.
The \emph{expansion ratio} of a graph $G$ is
\[
h(G):=\min_{\substack{W\subseteq V(G)\\0<|W|\le|G|/2}}\frac{|\partial(W)|}{|W|}.
\]
The \emph{expansion ratio} of a family 
$\CC$ of graphs is
\[
h(\CC)=\inf_{G\in\CC}h(G).
\]
If $\CC$ is infinite and $h(\CC)>0$ we call $\CC$ a \emph{family of expander graphs}
(Typically, we only use this terminology if $\CC$ is infinite.)

\begin{fact}[Folklore]
  For every $d$ there exists a family of $d$-regular $d$-connected
  expander graphs.
\end{fact}

Maybe the easiest way to obtain such a family is by taking random $d$-regular
bipartite graphs with both parts of the same size, which
asymptotically almost surely are $d$-connected \cite{ell83} and
have positive expansion \cite{alo86}.

Recll that a graph $G$ is \emph{$d$-connected} if $|V(G)|>d$ and for
every set $S\subseteq V(G)$ the graph $G\setminus S$ is connected. We
are only interested in 2-connected graphs here. It is well known that
every graph has a nice decomposition into its $2$-connected
components. It is convenient to state this result using tree
decompositions. A \emph{tree decomposition} of a graph $G$ is a pair
$(T,\beta)$, where $T$ is a tree and $\beta:V(T)\to 2^{V(G)}$ such
that: (i) for every $v\in V(G)$ the set of all $t\in V(T)$ such that
$v\in\beta(t)$ is connected in $T$, and (ii) for every edge $vw\in
E(G)$ there is a $t\in V(T)$ such that $v,w\in\beta(t)$. The
\emph{adhesion} of a tree decomposition $(T,\beta)$ is $\max_{tu\in
  E(T)}|\beta(t)\cap\beta(u)|$ if $E(T)\neq\emptyset$ and $0$ if $E(T)=\emptyset$.

\begin{fact}[Folklore]
  Every graph $G$ has a tree decomposition $(T,\beta)$ of adhesion at
  most $1$ such that for all $t\in V(T)$, either the induced subgraph
  $G[\beta(t)]$ is 2-connected or $|\beta(t)|\le 2$.
\end{fact}

We call the decomposition $(T,\beta)$ of the fact a
\emph{decomposition of $G$ into 2-connected components}.

\begin{lem}
  Let $\CE$ be a family of $3$-regular 2-connected expander graphs.

  Then there is  constant $c>0$ such that for every $G\in\CE$ and
  every set $W\subseteq V(G)$ there is a set
  $\hat W\supseteq W$ of size $|\hat W|\le c|W|$ such that $G\setminus
  \hat W$ is either empty or 2-connected. 
\end{lem}

\begin{proof}
  Let $\epsilon:=\min\{1,h(\CE)\}$ and 
  \[
  c:=\frac{30}{\epsilon}.
  \]
  Let $n:=|V(G)$, and let $W\subseteq V(G)$ and $k:=|W|$. Without loss
  of generality we may assume that
  \begin{equation}
    \label{eq:3}
    \frac{30}{\epsilon}k< n;
  \end{equation}
  otherwise we let $\hat W:=V(G)$.

  Let $(T,\beta)$ be a tree decomposition of
  $G\setminus W$ into 2-connected components. For every edge
  $tu\in E(T)$, we let $T(t,u)$ be the connected component of
  $T-\{tu\}$ (the tree obtaineed from $T$ by deleting the edge $tu$)
  that contains $u$, and we let $\gamma(t,u):=\bigcup_{s\in
    V(T(t,u))}\beta(s)$. We define $T(u,t)$ and $\gamma(u,t)$
  similarly. 

  Now we orient every edge $tu$ in such a way that it points to the
  larger of the two sets $\gamma(t,u)$ and $\gamma(u,t)$, breaking
  ties arbitrarily. Then there
  is a node $s\in V(T)$ such that all edge $st$ are oriented towards
  $s$. That is, for all $t\in N(s)$ (the set of neighbours of $s$ in $T$) we have
  $|\gamma(s,t)|\le|\gamma(t,s)|$. For every $t\in N(s)$, we let
  $\alpha(t):=\gamma(s,t)\setminus\beta(s)$. Note that 
  \begin{equation}
    \label{eq:1}
    |\alpha(t)|\le \frac{|V(G)\setminus W|}{2}=\frac{n-k}{2}
  \end{equation}
  Without loss of generality we assume that $\alpha(t)\neq\emptyset$
  for all $t\in N(s)$.

  Suppose for contradiction that
  $|\beta(s)|<3$. Let $W'=W\cup\beta(t)$. It follows from \eqref{eq:1}
  that there is a partition $(X,Y)$
  of $V(G)\setminus W'=\bigcup_{t\in N(s)}\alpha(t)$ such that there
  is no edge from $X$ to $Y$ in $G$ and 
  \[
  \frac{|V(G)\setminus W'|}{3}\le|X|\le|Y|\le \frac{2|V(G)\setminus
    W'|}{3}
  \]
  (both $X$ and $Y$ are unions of sets $\alpha(t)$ for $t\in N(s)$).
  Then $|X|\le
  n/2$ and thus 
  \[
  3|W'|\ge\partial(X)\ge \epsilon|X|\ge \frac{\epsilon|V(G)\setminus
    W'|}{3}=\frac{\epsilon}{3}(n-|W'|).
  \]
  This implies 
  \[
  \frac{10}{\epsilon}(k+2)\ge\left(\frac{9}{\epsilon}+1\right)|W'|\ge
  n,
  \]
  which
  contradicts \eqref{eq:3}. Thus $|\beta(s)|\ge 3$, and this means that
  $G[\beta(t)]$ is 2-connected.

  Next, we observe that for every $t\in N(s)$ there is at most one
  edge $e=vw\in E(G)$ such that $v\in\alpha(t)$ and $w\in
  V(G)\setminus(W\cup\alpha(t))$. To see this, suppose for
  contradiction that there are two such edges $v_1w_1$ and
  $v_2w_2$. Then $w_1=w_2=:w$ is the unique vertex in
  $\beta(s)\cap\beta(t)$, and therefore $v_1\neq v_2$. As $G[\beta(s)]$ is 2-connected, $G$ has at
  least two neighbours in $\beta(s)$. But then the degree of $w$ is at
  least $4$, which contradicts $G$ being 3-regular. 

  Hence
  \begin{equation}
    \label{eq:20}
    \epsilon|\alpha(t)|\le|\partial(\alpha(t))|\le 1+e(\alpha(t),W),
  \end{equation}
  where $e(\alpha(t),W)$ is the number of edges between $\alpha(t)$ and $W$.
  Moreover, for every $t\in N(W)$ we have $e(\alpha(t),W)\ge 1$,
  because otherwise the set $\beta(s)\cap\beta(t)$ of size at most $1$
  separates $G$, which contradicts $G$ being 2-connected. Note that
  here we use the assumption $\alpha(t)\neq\emptyset$.

  As $|\partial (W)|\le 3k$, it follows that $|N(s)|\le 3k$. Then
  \begin{align*}
  \Big|\bigcup_{t\in N(s)}\alpha(t)\Big|&=\sum_{t\in
    N(s)}|\alpha(t)|\\
    &\le\frac{|N(s)|+\sum_{t\in N(s)}e(\alpha(t),W)}{\epsilon}&\text{by
                                                               \eqref{eq:20}}\\
    &\le\frac{|N(s)|+|\partial(W)|}{\epsilon}\\
    &\le\frac{6k}{\epsilon}
  \end{align*}
  We let $\hat W:=W\cup\bigcup_{t\in N(s)}\alpha(t)$. Then
  $G\setminus \hat W=G[\beta(s)]$ is 2-connected, and 
  \[
  |\hat W|\le
  \Big(1+\frac{6}{\epsilon}\Big)k\le\frac{7}{\epsilon}k\le ck.
  \qedhere
  \]
\end{proof}

\begin{cor}\label{cor:exp}
  Let $\CE$ be a family of $3$-regular 2-connected expander graphs.

  There is constant $c>0$ such that for every $G\in\CE$ and every set
  $X\subseteq E(G)$ there is a set $X^*\supseteq X$ of size
  $|X^*|\le c|X|$ such that $E(G)\setminus X^*$ is either empty or the
  edge set of a 2-connected subgraph of $G$.
\end{cor}

\end{document}